\documentclass[aip,jmp]{revtex4-1}

\usepackage[nointlimits,nosumlimits]{amsmath}
\usepackage{amsfonts,amssymb,amsthm,ifthen}
\usepackage[utf8x]{inputenc}

\newtheorem{Thm}{Theorem}[section]
\newtheorem{Prp}[Thm]{Proposition}
\newtheorem{Lem}[Thm]{Lemma}

\newtheorem{Def}[Thm]{Definition}

\newcommand{\id}{\mathrm{id}}
\newcommand{\dvol}{\mathrm{dvol}}
\newcommand{\tr}{\textrm{tr}}
\newcommand{\Spin}{\mathrm{Spin}}
\newcommand{\Hom}{\mathrm{Hom}}

\newcommand{\setsep}{\;\big|\;}
\newcommand{\dbar}{\overline{\partial}}

\newcommand{\Dirac}{\rlap {\hspace{-0.5mm} \slash} D}

\newcommand{\mA}{\mathcal A}

\newcommand{\mL}{\mathcal L}

\newcommand{\mN}{\mathcal N}
\newcommand{\mO}{\mathcal O}
\newcommand{\mR}{\mathcal R}
\newcommand{\mS}{\mathcal S}
\newcommand{\bC}{\mathbb C}
\newcommand{\bN}{\mathbb N}
\newcommand{\bR}{\mathbb R}
\newcommand{\bZ}{\mathbb Z}

\newcommand{\os}{\overline{s}}

\newcommand{\ox}{\overline{x}}
\newcommand{\oy}{\overline{y}}
\newcommand{\oz}{\overline{z}}

\renewcommand\Re{\mathrm{Re}}

\newcommand{\dd}[2][]{\frac{\partial #1}{\partial #2}}
\newcommand{\abs}[1]{\left| #1 \right|}

\newcommand{\arr}[2]{%
  \begin{array}{@{}#1@{}}#2\end{array}}
\newcommand{\scal}[3][]{\ifthenelse{\equal{#1}{}}{
  \left\langle #2,\,#3 \right\rangle
}{\ifthenelse{\equal{#1}{(}}{
  \left( #2,\,#3 \right)
}{\ifthenelse{\equal{#1}{[}}{
  \left[ #2,\,#3 \right]
}{
  #1\left( #2,\,#3 \right)
}}}}

\begin{document}

\title{Holomorphic Supercurves and Supersymmetric Sigma Models}
\author{Josua Groeger}
\email{groegerj@mathematik.hu-berlin.de}
\affiliation{Humboldt-Universit\"at zu Berlin, Institut f\"ur Mathematik,\\
Rudower Chaussee 25, 12489 Berlin, Germany}
\date{\today}

\begin{abstract}
We introduce a natural generalisation of holomorphic curves to morphisms of supermanifolds,
referred to as holomorphic supercurves. More precisely, supercurves are morphisms
from a Riemann surface, endowed with the structure of a supermanifold which is
induced by a holomorphic line bundle, to an ordinary almost complex manifold.
They are called holomorphic if a generalised Cauchy-Riemann condition is satisfied.
We show, by means of an action identity, that holomorphic supercurves are special extrema of a
supersymmetric action functional.
\end{abstract}

\maketitle

\section{Introduction}

It is well-known that holomorphic curves are special extrema of the harmonic
action functional (Refs.~\onlinecite{MS04,Xin96,Jos02}).
Let us briefly recall the relevant background.
Let $\Sigma$ be a closed Riemann surface with a fixed complex structure $j$.
Moreover, for simplicity, we assume that $\Sigma$ is connected and fix a Riemann metric $h$
in the conformal class corresponding to $j$.
Let $X$ denote a smooth manifold of dimension $2n$. With respect to an almost complex
structure $J$ on $X$, a smooth map $\varphi:\Sigma\rightarrow X$
is called a ($J$-)holomorphic curve if
$\dbar_J\varphi:=\frac{1}{2}(d\varphi+J\circ d\varphi\circ j)=0$ vanishes.
If, on the other hand, $g$ is a Riemann metric on $X$,
we let $\mA(\varphi)$ denote the harmonic action functional,
\begin{align}
\label{eqnHarmonicAction}
\mA(\varphi):=\frac{1}{2}\int_{\Sigma}\abs{d\varphi}_{h,g}^2\dvol_{\Sigma}
\end{align}
$\mA$ is conformally invariant. Critical points, which are called harmonic,
are precisely those maps $\varphi$ whose tension field
$\tau(\varphi):=\tr\left((\nabla_{\cdot}d\varphi)(\cdot)\right)\in\Gamma(\varphi^*TX)$
vanishes. Now, let $\omega$ be a symplectic form on $X$ and $J$ be an almost complex
structure which is $\omega$-compatible and thus determines a Riemann metric $g_J$ by prescribing
$\scal[g_J]{v}{w}:=\scal[\omega]{v}{Jw}$.
Then every smooth map $\varphi:\Sigma\rightarrow X$ satisfies the action identity
\begin{align}
\label{eqnEnergyIdentity}
\frac{1}{2}\int_{\Sigma}\abs{d\varphi}_{h,g_J}^2\dvol_{\Sigma}
=\int_{\Sigma}\varphi^*\omega+\int_{\Sigma}\abs{\dbar_J\varphi}_{h,g_J}^2\dvol_{\Sigma}
\end{align}
It follows that $J$-holomorphic curves minimise the energy within their homology
class $[\varphi]:=\varphi_*([\Sigma])\in H_2(X,\bZ)$
and, as a consequence, every $J$-holomorphic curve is harmonic.

In this article, we introduce holomorphic supercurves as a natural generalisation of
holomorphic curves and show that they are special extrema of an action functional
$\mA_1$ which extends $\mA$. It is organised as follows.
In Sec. \ref{secHolomorphicSupercurvesAdhoc}, we give an ad-hoc definition of
holomorphic supercurves to be motivated later and construct non-trivial examples
in important cases.
In Sec. \ref{secSusySigmaModels}, we introduce the action functional $\mA_1$
as an extension of the harmonic action. We then compare it with another such extension,
known as the Dirac-harmonic action, calculating the Euler-Lagrange equations.
In Sec. \ref{secHolomorphicSupercurves}, we define holomorphic supercurves in the
natural context of supermanifolds. We show that these  (genuine) holomorphic supercurves may be
identified with holomorphic supermanifolds in the sense of the previous ad-hoc definition.
Finally, we state and prove in Sec. \ref{secSuperActionIdentity} a
generalisation of the classical action identity (\ref{eqnEnergyIdentity}), by means of
supermanifold theory. We show that, as a consequence,
every holomorphic supercurve extremises $\mA_1$, thus proving our main theorem.

\section{Holomorphic Supercurves}
\label{secHolomorphicSupercurvesAdhoc}

Holomorphic supercurves generalise holomorphic curves as described in the
following ad-hoc definition.
For smooth maps $\varphi:\Sigma\rightarrow X$, we define the operator
$D_{\varphi}:\Omega^0(\Sigma,\varphi^*TX)\rightarrow\Omega^{0,1}(\Sigma,\varphi^*TX)$ by
\begin{align}
\label{eqnLinearisedDbar}
D_{\varphi}\xi:=\frac{1}{2}(\nabla\xi+J(\varphi)\nabla\xi\circ j)
-\frac{1}{2}J(\varphi)(\nabla_{\xi}J)(\varphi)\partial_J(\varphi)
\end{align}
where $\nabla:=\nabla^{g_J}$ denotes the Levi-Civita connection of $g_J$ (and its pullback under $\varphi$).
It can also be expressed in the form
\begin{align}
\label{eqnLinearisedDbarCompatible}
D_{\varphi}\xi=(\nabla^J\xi)^{0,1}+\frac{1}{4}N_J(\xi,\partial_J\varphi)
\end{align}
where $N_J$ denotes the Nijenhuis tensor, $\nabla^J$ the $J$-complexification, and $(\nabla^J)^{0,1}$
the $(0,1)$-part of the $1$-form $\nabla^J$.
If $\varphi$ is holomorphic, (\ref{eqnLinearisedDbar}) and (\ref{eqnLinearisedDbarCompatible})
coincide with the vertical differential (the linearisation) of $\dbar_J$ at $\varphi$.
Let $D_{\varphi}^{\bC}$ and $N_J^{\bC}$ denote the complex-linear extensions
of $D_{\varphi}$ and $N_J$, respectively.

\begin{Def}
\label{defHolomorphicSupercurveAdhoc}
Let $L\rightarrow\Sigma$ be a holomorphic line bundle and $J$ be an almost complex structure on $X$.
Then a \emph{holomorphic supercurve} is a tuple $(\varphi,\psi_1,\psi_2,\xi)$,
consisting of a smooth map $\varphi:\Sigma\rightarrow X$ and sections
$\psi_1,\psi_2\in\Gamma(\Sigma,L\otimes_J\varphi^*TX)$ and $\xi\in\Gamma(\Sigma,\varphi^*T^{\bC}X)$
such that
\begin{align*}
N_J^{\bC}(\psi_{1\theta},\psi_{2\theta})=0\,,\quad\dbar_J\varphi=0\,,\quad
D_{\varphi}^{\bC}\xi=0\,,\quad D_{\varphi}^{\bC}\psi_{1\theta}=0\,,\quad D_{\varphi}^{\bC}\psi_{2\theta}=0
\end{align*}
where, for $U\subseteq\Sigma$ sufficiently small, we fix a nonvanishing holomorphic section
$\theta\in\Gamma(U,L)$ and let $\psi_{j\theta}\in\Gamma(U,\varphi^*T^{1,0}X)$ be
such that $\psi_j=\theta\cdot\psi_{j\theta}$ holds for $j=1,2$.
\end{Def}

Upon identifying $\psi_{j\theta}$ with a (local) section of $\varphi^*TX$,
the last two conditions may be reformulated into
$D_{\varphi}\psi_{j\theta}=D_{\varphi}(J\psi_{j\theta})=0$.
Since $D_{\varphi }$ is a real linear Cauchy-Riemann operator, it is clear that the
conditions stated do not depend on the choice of $\theta\in\Gamma(U,L)$,
provided that $\theta$ is chosen holomorphic.
Well-definedness also follows from the characterisation of holomorphic supercurves thus defined
with natural (global) objects in terms of supergeometry (Prp. \ref{prpHolomorphicSupercurveLocal}).
It is possible to perturb the defining equations for holomorphic supercurves
by making them depend on a connection $A$ such that the resulting moduli spaces are,
generically, finite dimensional manifolds and such that, in important cases, every sequence
has a convergent subsequence provided that a suitable extension of the
classical energy is uniformly bounded, a version of Gromov compactness.
These objects are called $(A,J)$-holomorphic supercurves and will be thoroughly examined in upcoming papers.

We shall next provide a class of examples of non-trivial holomorphic supercurves for the case $L$
being a half spinor bundle, using a construction similar to an example in Ref.~\onlinecite{CJLW06}.
To fix notation, let us first recall some basics of $2$-dimensional spin geometry
(Refs.~\onlinecite{Fri00,Bau81,Hij01}).
The standard representation $\Gamma$ on $\bC^2$ of the complexified Clifford algebra
$Cl_2^{\bC}:=Cl_2\otimes\bC$ is determined by the gamma matrices which, in our convention,
are given by
\begin{align*}
\Gamma_{e_1}=\left(\arr{cc}{0&-1\\1&0}\right)\,,\;
\Gamma_{e_2}=\left(\arr{cc}{0&i\\i&0}\right)\,,\quad\mathrm{whence}\quad
\Gamma_{e_1}\Gamma_{e_2}=\Gamma_{e_1e_2}=\left(\arr{cc}{-i&0\\0&i}\right)
\end{align*}
where $e_1,e_2$ denotes the standard basis of $\bR^2$.
The induced action $\mu:=\Gamma_{|\Spin(2)}$ on $\bC^2$ of the spin group
$\Spin(2)=\{a\cdot 1+b\cdot e_1 e_2\setsep a^2+b^2=1\}\cong S^1\subseteq\bC$ is then
\begin{align*}
\mu(s)\left(\arr{c}{x\\y}\right)=\mu(a\cdot 1+b\cdot e_1e_2)\left(\arr{c}{x\\y}\right)
=\left(\arr{c}{(a-ib)\cdot x\\(a+ib)\cdot y}\right)
=\left(\arr{c}{\os\cdot x\\s\cdot y}\right)
=:\left(\arr{c}{\mu_+(s)(x)\\\mu_-(s)(y)}\right)
\end{align*}
with the two irreducible $\Spin(2)$ representations $\mu_+$ and $\mu_-$ such that
$\mu=\mu_+\oplus\mu_-$.
Fixing a spin structure $\Spin(\Sigma)$ on $(\Sigma,h)$,
we define the (half-)spinor bundles
\begin{align*}
S^{\pm}:=\Spin(\Sigma)\times_{(\Spin(2),\mu_{\pm})}\bC\;,\qquad
S:=\Spin(\Sigma)\times_{(\Spin(2),\mu)}\bC^2=S^+\oplus S^-
\end{align*}
carrying natural holomorphic structures induced by the canonical isomorphisms\\
$(S^+)^2:=S^+\otimes_{\bC}S^+\cong \Lambda^{1,0}\Sigma$ and $(S^-)^2\cong T^{1,0}\Sigma$
of complex line bundles (Ref.~\onlinecite{Ati71}).

Let $L$ be a complex line bundle and $d\in\bZ^*$. There is a canonical bundle map
$l_d:L\rightarrow L^d$ induced by the corresponding power function in compatible trivialisations of $L$ and $L^d$.
We use the following terminology. Let $v$ be a tangent vector of $\Sigma$. Then either preimage
under the bundle map $l_2:S^-\rightarrow (S^-)^2\cong T^{1,0}\Sigma$
is called a square root of $v$, and analogous with $S^-$ and $T^{1,0}\Sigma$ replaced
by $S^+$ and $\Lambda^{1,0}\Sigma$, respectively.

\begin{Def}
\label{defSquareRootCoordinates}
Let $z=s+it$ be holomorphic coordinates on $\Sigma$ and denote by $\partial_z:=\dd{z}$ and
$e_z:=\lambda^{-\frac{1}{2}}\partial_z$ the induced local vector fields in $T^{1,0}\Sigma$
where $\lambda=\scal[h]{\partial_s}{\partial_s}$ is the conformal factor of the metric.
We denote by $\theta^-$ either square root of $\partial_z$
and by $\theta^+:=l_{-1}(\theta^-)$ its image under $l_{-1}:S^-\rightarrow (S^-)^{-1}\cong S^+$.
Moreover, we define $e^+:=\lambda^{\frac{1}{4}}\theta^+$ and $e^-:=\lambda^{-\frac{1}{4}}\theta^-$.
\end{Def}

By definition, $\theta^+$ is a square root of $dz$ since $dz=l_{-1}(\partial_z)$ holds
and, moreover, $\theta^{\pm}$ are nonvanishing (local) holomorphic
sections of $S^{\pm}$. Similarly, $e^-$ is a square root of $e_z$ and $e^+=l_{-1}(e^-)$.
Let $s$ be a local section of $\Spin(\Sigma)$ such that we may identify
$e_z=[s,1]\in T\Sigma|_U\cong\Spin(\Sigma)|_U\times_{(\mu_-)^2}\bC$.
Then, upon replacing $s$ with $-s$ if necessary, we may further identify
$e^+=[s,(1,0)],\,e^-=[s,(0,1)]\in\Spin(\Sigma)\times_{\mu}\bC^2\cong S$.

We need the following construction.
Let $\varphi\in C^{\infty}(\Sigma,X)$ be a smooth map and $\zeta\in\Gamma(\Sigma,S^-)$ be a section
of $S^-$, (locally) identified with a complex valued function $\zeta_-$ via
$\zeta=\zeta_-\cdot\theta^-$. Upon complex linear extension $d\varphi[iv]:=id\varphi[v]\in T^{\bC}X$
for $iv\in T^{\bC}\Sigma$, $\varphi$ and $\zeta$ together induce a section
\begin{align}
\label{eqnInducedSection}
\psi^{\varphi,\zeta}:=\zeta_-\theta^+\cdot d\varphi[\partial_z]\in\Gamma(S^+\otimes_{\bC}\varphi^*T^{\bC}X)
\end{align}
Note that, due to the transformation behaviour of holomorphic coordinates and the induced sections of $S^{\pm}$,
this definition is independent thereof and $\psi^{\varphi,\zeta}$ is indeed a global section.
If $\varphi$ is $J$-holomorphic, then $i\cdot\psi^{\varphi,\zeta}=J\psi^{\varphi,\zeta}$ follows.
In other words, $\dbar_J\varphi=0$ implies $\psi^{\varphi,\zeta}\in\Gamma(S^+\otimes_J\varphi^*TX)$.

\begin{Lem}
Let $(X,\omega)$ be a symplectic manifold and $J$ be an $\omega$-compatible almost complex structure.
Let $\varphi\in C^{\infty}(\Sigma,X)$ be a $J$-holomorphic curve, 
$\zeta_1,\zeta_2\in\Gamma(S^-)$ be holomorphic spinors and $\xi\in\Gamma(\Sigma,\varphi^*T^{\bC}X)$
be such that $D_{\varphi}^{\bC}\xi=0$.
Then $(\varphi,\psi^{\varphi,\zeta_1},\psi^{\varphi,\zeta_2},\xi)$ is a holomorphic supercurve.
\end{Lem}

Since the space of holomorphic sections of $S^-$ has dimension $1-g+\deg S^-=2-g$
by the classical Riemann-Roch theorem (Ref.~\onlinecite{Gas06}),
it follows in particular that non-trivial holomorphic supercurves do exist for $\Sigma$
having genus $g=0$ or $g=1$ (sphere or torus, respectively),
provided that a non-trivial holomorphic curve $\varphi$ exists.

\begin{proof}
We set $\psi_j:=\psi^{\varphi,\zeta_j}$ for $j=1,2$.
Then, by construction, $\psi_{j\theta^+}=\zeta_{j-}\cdot d\varphi[\partial_z]$.
Among the conditions of Def. \ref{defHolomorphicSupercurveAdhoc},
it remains to show that $\scal[N_J^{\bC}]{\psi_{1\theta^+}}{\psi_{2\theta^+}}=0$ and
$D_{\varphi}^{\bC}\psi_{j\theta^+}=0$.
The second condition can be reformulated into $D_{\varphi}\psi_{j\theta^+}=D_{\varphi}(J\psi_{j\theta^+})=0$
or, equivalently, since $D_{\varphi}$ is the sum of a complex linear and a complex antilinear operator
by (\ref{eqnLinearisedDbarCompatible}), into
$(\nabla^J\psi_{j\theta^+})^{0,1}=\scal[N_J]{\psi_{j\theta^+}}{\partial_J\varphi}=0$. Since $\zeta_{j-}$
and $\varphi$ are holomorphic and $N_J$ is complex antilinear with respect to $J$,
it thus suffices to prove that
\begin{align}
\label{eqnNablaVarphi}
\scal[N_J^{\bC}]{d\varphi[\partial_s]}{d\varphi[\partial_s]}=0\;,\qquad
(\nabla^J_{\partial_s}d\varphi[\partial_s])^{0,1}=0
\end{align}
in order to show that $(\varphi,\psi_1,\psi_2,\xi)$ is a holomorphic supercurve.
The first condition is void as $N_J$ is skew-symmetric,
and it remains to show the second one. We calculate
\begin{align*}
4(\nabla^J_{\partial_s}d\varphi[\partial_s])^{0,1}
&=2\left(\nabla^J_{\partial_s}d\varphi[\partial_s]+J\nabla^J_{j[\partial_s]}d\varphi[\partial_s]\right)\\
&=\nabla_{\partial_s}d\varphi[\partial_s]-J\nabla_{\partial_s}(Jd\varphi[\partial_s])
+J\nabla_{\partial_t}d\varphi[\partial_s]+\nabla_{\partial_t}(Jd\varphi[\partial_s])\\
&=\left(\nabla_{\partial_s}d\varphi[\partial_s]+\nabla_{\partial_t} d\varphi[\partial_t] \right)
+J\left(\nabla_{\partial_t}d\varphi[\partial_s]-\nabla_{\partial_s}d\varphi[\partial_t]\right)
\end{align*}
For $a,b\in\{s,t\}$, we obtain the local formula
\begin{align*}
\nabla_{\partial_a} d\varphi[\partial_b]&=\nabla_{\partial_a}\left(\dd[\varphi^m]{b}\dd{x^m}\right)
=\partial_a\left(\dd[\varphi^k]{b}\right)\dd{x^k}+\dd[\varphi^m]{b}\nabla_{d\varphi[\partial_a]}\dd{x^m}\\
&=\frac{\partial^2\varphi^k}{\partial_a\partial_b}\dd{x^k}+\dd[\varphi^m]{b}\dd[\varphi^l]{a}
\nabla_{\dd{x^l}}\dd{x^m}\\
&=\left(\frac{\partial^2\varphi^k}{\partial_a\partial_b}
+\dd[\varphi^m]{b}\dd[\varphi^l]{a}\Gamma_{lm}^k\right)\dd{x^k}
\end{align*}
which is symmetric in $a,b$. Therefore, the second bracket in the
previous expression for $4(\nabla^Jd\varphi[\partial_s])^{0,1}$ vanishes,
while the first bracket becomes
\begin{align*}
&\nabla_{\partial_s} d\varphi[\partial_s]+\nabla_{\partial_t} d\varphi[\partial_t]\\
&\qquad\qquad=\left(\frac{\partial^2\varphi^k}{\partial_s\partial_s}
+\dd[\varphi^m]{s}\dd[\varphi^l]{s}\Gamma_{lm}^k\right)\dd{x^m}
+\left(\frac{\partial^2\varphi^k}{\partial_t\partial_t}
+\dd[\varphi^m]{t}\dd[\varphi^l]{t}\Gamma_{lm}^k\right)\dd{x^m}\\
&\qquad\qquad=\delta^{ij}\left(\frac{\partial^2\varphi^k}{\partial_i\partial_j}
+\dd[\varphi^l]{i}\dd[\varphi^m]{j}\Gamma_{lm}^k\right)\dd{x^m}
=\lambda\cdot\tau(\varphi)
\end{align*}
where $\lambda$ is the conformal factor of the metric $h$ and
we have used the following local formula for the tension field, which is obtained
by a standard calculation involving (1.2.9) in Ref.~\onlinecite{Xin96}.
\begin{align}
\label{eqnLocalTension}
\tau(\varphi)=\frac{4}{\lambda}\left(\frac{\partial^2\varphi^k}{\partial z\partial\oz}
+\Gamma^k_{lm}\dd[\varphi^l]{z}\dd[\varphi^m]{\oz}\right)\dd{x^k}
\end{align}
But $\tau(\varphi)$ vanishes since $\varphi$ is holomorphic and as such harmonic,
thus proving (\ref{eqnNablaVarphi}).
\end{proof}

\section{Variants of Supersymmetric Sigma Models}
\label{secSusySigmaModels}

In this section, we introduce and examine the action functional $\mA_1$.
Let us start with the mathematical framework.
With the notations from the previous sections, we assume that $\Sigma$ carries a spin structure
and let $g$ be a Riemann metric on $X$.
The $\bC$-bilinear form $B$ and the Hermitian metric $H$ on $\bC^2$, defined by
\begin{align*}
\scal[B]{\left(\arr{c}{x\\y}\right)}{\left(\arr{c}{z\\w}\right)}:=xw+yz\;,\qquad
\scal[H]{\left(\arr{c}{x\\y}\right)}{\left(\arr{c}{z\\w}\right)}:=\Re(\ox z+\oy w)
\end{align*}
are both invariant under the action $\mu$ of $\Spin(2)$ and thus induce on $S$ a well-defined bilinear form
and a Hermitian metric, respectively, in the following denoted by the same symbols.
Existence and properties of $B$ also follow from more general considerations
(Ref.~\onlinecite{Var04}, P.242, and Refs.~\onlinecite{Har90,Hij01}).
Now let $\varphi:\Sigma\rightarrow X$ be a smooth map. We consider the tensor product
$S\otimes\varphi^*TX:=S\otimes_{\bR}\varphi^*TX\cong S\otimes_{\bC}\varphi^*T^{\bC}X$.
Denoting the $\bC$-bilinear extension of the metric $g$
by the same symbol, $B$ and $g$ together induce a $\bC$-bilinear form on
$S\otimes_{\bC}\varphi^*T^{\bC}X\cong\Hom_{\bC}(S^*,\varphi^*T^{\bC}X)$ that we shall also denote by $B$.
Similarly, $g$ and the Hermitian metric $H$ on $S$ together induce a Hermitian bundle metric
on $S\otimes_{\bR}\varphi^*TX$, which we shall denote $\scal[(]{\cdot}{\cdot}$.
For later calculations, we need local expressions.
Let $z=s+it$ be holomorphic coordinates on $\Sigma$ and $\theta^{\pm}$ and $e^{\pm}$
be square roots of the induced (co)tangent fields as in Def. \ref{defSquareRootCoordinates}.
Any section $\psi\in\Gamma(S\otimes_{\bC}\varphi^*T^{\bC}X)$ may then be (locally) written
\begin{align}
\label{eqnLocalPsi}
\psi=\psi_+\theta^++\psi_-\theta^-=\psi_{e^+}e^++\psi_{e^-}e^-=:\left(\arr{c}{\psi_{e^+}\\\psi_{e^-}}\right)
\end{align}
and we obtain
\begin{align}
\label{eqnB}
\scal[B]{\psi}{\psi'}&=\scal[B]{\psi_{e^a}e^a}{\psi'_{e^b}e^b}
=B^{ab}\scal[g]{\psi_{e^a}}{\psi'_{e^b}}=\scal[g]{\psi_{e^+}}{\psi'_{e^-}}+\scal[g]{\psi_{e^-}}{\psi'_{e^+}}\\
\label{eqnH}
\scal[(]{\psi}{\psi'}&=\scal[(]{\psi_{e^A}e^A}{\psi'_{e^B}e^B}
=\scal[g]{\psi_{e^A}}{\psi'_{e^B}}\cdot\scal[H]{e^A}{e^B}
\end{align}
where we use Einstein's summation convention and, in the first line,
$a,b\in\{+,-\}$ and $B^{ab}$ denotes the matrix elements of the inverse of $B$ in the coordinates
determined by $e^{\pm}$, whereas in the second line, we write
$\psi=\psi_{e^A}e^A$ where $e^A\in\{e^+,ie^+,e^-,ie^-\}$ and $\psi_{e^A}\in\Gamma(U,\varphi^*TX)$.
We further need the (twisted) Dirac operator along $\varphi$ as defined by the concatenation
\begin{align*}
\Dirac:\Gamma(S\otimes_{\bR}\varphi^*TX)\stackrel{\nabla}{\rightarrow}\Gamma(T^*\Sigma\otimes_{\bR} S\otimes_{\bR}\varphi^*TX)
\stackrel{\Gamma}{\rightarrow}\Gamma(S\otimes_{\bR}\varphi^*TX)
\end{align*}
where $\nabla:=\nabla^S\otimes_{\bR}\id+\id\otimes_{\bR}\nabla^{\varphi^*TX}$ is the
induced connection on $S\otimes_{\bR}\varphi^*TX$.
Using $\nabla^S e^{\pm}=0$, we further obtain the following local expression.
\begin{align}
\label{eqnTwistedDiracOperatorLocal}
\Dirac\psi=\Dirac(\psi_{e^+}e^++\psi_{e^-}e^-)=\Dirac\left(\arr{c}{\psi_{e^+}\\\psi_{e^-}}\right)
=2\lambda^{-\frac{1}{2}}\left(\arr{c}{-\nabla_{\dd{z}}\psi_{e^-}\\\nabla_{\dd{\oz}}\psi_{e^+}}\right)
\end{align}
where $\lambda$ is the conformal factor of the metric.
Finally, we use the following conventions. In general, let $\bigwedge\bC^N$ denote the complex Graßmann algebra
(also called exterior algebra) with $N$ generators $\eta^1,\ldots,\eta^N$, which we fix in the following
(this can be any basis of $\bC^N$), and let $\bigwedge\bC^N_{\mathrm{odd}}$ and
$\bigwedge\bC^N_{\mathrm{even}}$ denote the respective subspaces of odd and even forms.
Any $\bC$-bilinear form $\scal[(]{\cdot}{\cdot}$ on a complex
vector space $V$ can be extended to a $\bigwedge\bC^N$-bilinear form on $V\otimes_{\bC}\bigwedge\bC^N$
by prescribing
\begin{align*}
\scal[(]{\eta^{i_1}\cdot\ldots\cdot\eta^{i_l}\cdot v}{\eta^{j_1}\cdot\ldots\cdot\eta^{j_k}\cdot w}
:=\eta^{i_1}\cdot\ldots\cdot\eta^{i_l}\cdot\eta^{j_1}\cdot\ldots\cdot\eta^{j_k}\cdot\scal[(]{v}{w}
\end{align*}
for $v,w\in V$ and abbreviating $\eta^1\cdot\eta^2:=\eta^1\wedge\eta^2$ for the wedge product.
Similarly, we define the integral over a function $f$ with values in $\bigwedge\bC^N$ by
\begin{align*}
\int\left(\eta^{i_1}\cdot\ldots\cdot\eta^{i_l}\cdot f_i
+\eta^{j_1}\cdot\ldots\cdot\eta^{j_k}\cdot f_j\right)
:=\left(\eta^{i_1}\cdot\ldots\cdot\eta^{i_l}\cdot\int f_i\right)
+\left(\eta^{j_1}\cdot\ldots\cdot\eta^{j_k}\cdot\int f_j\right)
\end{align*}
In the following, we restrict ourselves to the case $N=2$, and we fix two
generators $\eta^1$ and $\eta^2$ of $\bigwedge\bC^2$.

\begin{Def}
\label{defSupersymmetricAction}
For $\varphi\in C^{\infty}(\Sigma,X)$, $\xi\in\Gamma(\Sigma,\varphi^*T^{\bC}X)$ and
\begin{align*}
\psi=\eta^1\psi_1+\eta^2\psi_2
\in\Gamma\left(\Sigma,S^+\otimes_{\bC}\varphi^*T^{\bC}X\otimes\bigwedge\bC^2_{\mathrm{odd}}\right)
\end{align*}
we denote by $\mA_1$ the functional
\begin{align*}
\mA_1(\varphi,\psi_1,\psi_2,\xi):=\frac{1}{2}\int_{\Sigma}\dvol_{\Sigma}\left(\abs{d\varphi}^2
-2\eta^1\eta^2\,\scal[g]{\xi}{\tau(\varphi)}-2i\,\scal[B]{\psi}{\Dirac\psi}\right)
\in\bigwedge\bC^2_{\mathrm{even}}
\end{align*}
where $\tau(\varphi)$ is the tension field of $\varphi$, $B$ is the complex
bilinear form from (\ref{eqnB}) and $\Dirac$ denotes the twisted Dirac operator,
restricted to sections of $S^+\otimes\varphi^*TX\subseteq S\otimes\varphi^*TX$.
\end{Def}

Being Graßmann valued, the functional $\mA_1$ may be considered as consisting of two integrals:
one which is proportional to $1$ and one which is proportional to $\eta^1\eta^2$.
$\mA_1$ is an extension of the harmonic action $\mA$ in (\ref{eqnHarmonicAction}), and
its critical points are sometimes called superharmonic maps (Ref.~\onlinecite{Khe05}).
$\mA_1$ is known as $\mN=(1,0)$ supersymmetric sigma model in the physical literature
(Refs.~\onlinecite{DF99b,Wit92}). $\mA_1$ remains well-defined upon replacing $S^+$ by the
full spinor bundle $S=S^+\oplus S^-$. Further adding an additional curvature term
makes it the so called $\mN=(1,1)$-model.
For comparison with the literature note that, by the proof of Lem. \ref{lemNonsuperAction}
below, the first two terms may be rewritten
\begin{align*}
\frac{1}{2}\int_{\Sigma}\dvol_{\Sigma}\left(\abs{d\varphi}^2
-2\eta^1\eta^2\,\scal[g]{\xi}{\tau(\varphi)}\right)
=i\int_{\Sigma}dz\wedge d\oz\;\phi_0(g_{ij})\dd[\phi_0^i]{z}\dd[\phi_0^j]{\oz}
\end{align*}
for some $\bigwedge\bC^2_{\mathrm{even}}$-valued field $\phi_0$ in which the even field $\xi$
is implicitly contained.

\begin{Thm}
\label{thmSupersymmetricAction}
Let $(X,\omega)$ be a symplectic manifold, $J$ be an $\omega$-compatible almost complex structure
and $g$ be the induced Riemann metric.
Then the action functional $\mA_1$ satisfies the following implication.
If the tuple $(\varphi,\psi_1,\psi_2,\xi)$ is a holomorphic supercurve in the sense
of Def. \ref{defHolomorphicSupercurveAdhoc} with $L=S^+$, then it is a critical point of $\mA_1$.
\end{Thm}

This is our main theorem. We defer its proof to the very end of this article and,
for the time being, turn to the Euler-Lagrange equations. Analogous to (\ref{eqnLocalPsi}),
we shall use local expressions such as
$\psi=\psi_+\theta^+=(\eta^1\psi_{1+}+\eta^2\psi_{2+})\theta^+$ for $\psi$ as in Def. \ref{defSupersymmetricAction}.

\begin{Prp}
\label{prpEulerLagrange}
The tuple $(\varphi,\psi_1,\psi_2,\xi)$ is a critical point of $\mA_1$ if and only if
the following (Euler-Lagrange) equations are satisfied.
\begin{align*}
&\tau(\varphi)=0\;,\qquad\nabla_{\dd{\oz}}\psi_{1+}=0\;,\qquad\nabla_{\dd{\oz}}\psi_{2+}=0\\
&R^{\nabla}(\psi_{1e_+},\psi_{2e_+})e_{\oz}(\varphi)=\nabla^2\xi+R^{\nabla}(e_s,\xi)e_s+R^{\nabla}(e_t,\xi)e_t
\end{align*}
\end{Prp}

It is clear that the Euler-Lagrange equations stated must be globally well-defined,
i.e. do not depend on the choice of holomorphic coordinates. As in (\ref{eqnInducedSection}),
this also follows directly from the transformation behaviour.

\begin{proof}
By the local formula for the twisted Dirac operator in (\ref{eqnTwistedDiracOperatorLocal})
and oddness of $\psi$, we first transform the term of $\mA_1$ involving $\Dirac\psi$ as follows.
\begin{align*}
\int_{\Sigma}\dvol_{\Sigma}\,\scal[B]{\psi}{\Dirac\psi}
&=\int_{\Sigma}\dvol_{\Sigma}\,2\lambda^{-\frac{1}{2}}\scal[g]{\psi_{e^+}}{\nabla_{\partial_{\oz}}\psi_{e^+}}\\
&=\int_{\Sigma}\dvol_{\Sigma}\,2\lambda^{-\frac{1}{2}}\scal[g]{\lambda^{-\frac{1}{4}}\psi_+}{\nabla_{\partial_{\oz}}(\lambda^{-\frac{1}{4}}\psi_+)}\\
&=\int_{\Sigma}\dvol_{\Sigma}\,2\lambda^{-1}\scal[g]{\psi_+}{\nabla_{\partial_{\oz}}\psi_+}\\
&=\int_{\Sigma}dz\wedge d\oz\,\scal[g]{\psi_+}{\nabla_{\partial_{\oz}}\psi_+}
\end{align*}
where we used in the third equation that $\scal[g]{\psi_+}{\psi_+}=0$ vanishes since $\psi_+$ is odd.
Now, for $\psi,\psi'\in\Gamma(\Sigma,S^+\otimes_{\bC}\varphi^*T^{\bC}X\otimes_{\bC}\bigwedge\bC^2_{\mathrm{odd}})$,
we define the (global) Graßmann valued $1$-form $\Omega:=-\scal[g]{\psi_+}{\psi'_+}dz$.
Since $\Sigma$ is closed, we obtain
\begin{align*}
0&=\int_{\Sigma}d\Omega=\int_{\Sigma}dz\wedge d\oz\,\dd{\oz}\scal[g]{\psi_+}{\psi'_+}\\
&=\int_{\Sigma}dz\wedge d\oz\,\left(\scal[g]{\nabla_{\dd{\oz}}\psi_+}{\psi'_+}
+\scal[g]{\psi_+}{\nabla_{\dd{\oz}}\psi'_+}\right)
\end{align*}
After this preparation, consider a variation $\psi=\psi_{\varepsilon}$ with $\frac{d\psi}{d\varepsilon}=\gamma$
at $\varepsilon=0$ and fix $(\varphi,\xi)$. Then, using the previous calculations, we yield
\begin{align*}
i\cdot\frac{d\mA_1}{d\varepsilon}|_0
&=\frac{d}{d\varepsilon}|_0\int_{\Sigma}\dvol_{\Sigma}\,\scal[B]{\psi}{\Dirac\psi}\\
&=\int_{\Sigma}dz\wedge d\oz\,\frac{d}{d\varepsilon}|_0\,\scal[g]{\psi_+}{\nabla_{\dd{\oz}}\psi_+}\\
&=\int_{\Sigma}dz\wedge d\oz\,\left(\scal[g]{\gamma_+}{\nabla_{\dd{\oz}}\psi_+}+\scal[g]{\psi_+}{\nabla_{\dd{\oz}}\gamma_+}\right)\\
&=\int_{\Sigma}dz\wedge d\oz\,\left(\scal[g]{\gamma_+}{\nabla_{\dd{\oz}}\psi_+}-\scal[g]{\nabla_{\dd{\oz}}\psi_+}{\gamma_+}\right)\\
&=2\int_{\Sigma}dz\wedge d\oz\,\scal[g]{\gamma_+}{\nabla_{\dd{\oz}}\psi_+}
\end{align*}
This expression is required to vanish for any $\gamma$ and, therefore, $\nabla_{\dd{\oz}}\psi_+=0$ follows.
Next, we consider a variation $\xi=\xi_{\varepsilon}$ with $\frac{d\xi}{d\varepsilon}=\chi$ at
$\varepsilon=0$ and fix $(\varphi,\psi)$. Then
\begin{align*}
-\frac{d\mA_1}{d\varepsilon}|_0
=\eta^1\eta^2\int_{\Sigma}\dvol_{\Sigma}\,\frac{d}{d\varepsilon}|_0\,\scal[g]{\xi}{\tau(\varphi)}
=\eta^1\eta^2\int_{\Sigma}\dvol_{\Sigma}\,\scal[g]{\chi}{\tau(\varphi)}
\end{align*}
and hence we obtain $\tau(\varphi)=0$.
Next, we consider a variation $\varphi=\varphi_{\varepsilon}$
with $\frac{d\varphi}{d\varepsilon}=\zeta$ at $\varepsilon=0$ and fix $(\xi,\psi)$.
The terms in $\mA_1$ which are proportional to $\eta^1\eta^2$ are independent of those which are not and we may
treat both cases separately. Considering only the part not proportional to $\eta^1\eta^2$ reduces $\mA_1$ to the
classical action $\frac{1}{2}\int_{\Sigma}\dvol_{\Sigma}\,\abs{d\varphi}^2$ whose Euler-Lagrange equations
are known to be $\tau(\varphi)=0$, which we thus obtain a second time (Ref.~\onlinecite{Xin96}, Sec. 1.2.3).
Now we consider the part of $\mA_1$ which is proportional to $\eta^1\eta^2$, here denoted $\mA_1^{\eta}$.
We use the equations $\tau(\varphi)=0$ and $\nabla_{\dd{\oz}}\psi_+=0$ for critical points already established
to get
\begin{align*}
2\cdot\frac{d\mA_1^{\eta}}{d\varepsilon}|_0
&=-2\eta^1\eta^2\int_{\Sigma}\dvol_{\Sigma}\frac{d}{d\varepsilon}|_0\,\scal[g]{\xi}{\tau(\varphi)}
-2i\int_{\Sigma}dz\wedge d\oz\,\frac{d}{d\varepsilon}|_0\,\scal[g]{\psi_+}{\nabla_{\dd{\oz}}\psi_+}\\
&=-2\eta^1\eta^2\int_{\Sigma}\dvol_{\Sigma}\,\scal[g]{\xi}{\nabla_{\dd{\varepsilon}}\tau(\varphi)}
-2i\int_{\Sigma}dz\wedge d\oz\,\scal[g]{\psi_+}{\nabla_{\dd{\varepsilon}}\nabla_{\dd{\oz}}\psi_+}
\end{align*}
As for the first term, note that $\nabla_{\dd{\varepsilon}}\tau(\varphi)=\nabla^2\zeta+R^{\nabla}(e_i,\zeta)e_i$
by Ref.~\onlinecite{Xin96}, Sec. 1.4.3, where $\nabla^2\zeta:=\tr(\nabla_{\cdot}(\nabla\zeta))(\cdot)$
is the trace-Laplacian operator, which is formally self-adjoint. We thus calculate
\begin{align*}
\int_{\Sigma}\dvol_{\Sigma}\,\scal[g]{\xi}{\nabla_{\dd{\varepsilon}}\tau(\varphi)}
&=\int_{\Sigma}\dvol_{\Sigma}\,\scal[g]{\xi}{\nabla^2\zeta+R^{\nabla}(e_i,\zeta)e_i}\\
&=\int_{\Sigma}\dvol_{\Sigma}\,\scal[g]{\zeta}{\nabla^2\xi+R^{\nabla}(e_i,\xi)e_i}
\end{align*}
and, for the second term,
\begin{align*}
&-2i\int_{\Sigma}dz\wedge d\oz\,\scal[g]{\psi_+}{\nabla_{\dd{\varepsilon}}\nabla_{\dd{\oz}}\psi_+}\\
&\qquad\qquad=-2i\int_{\Sigma}dz\wedge d\oz\,\scal[g]{\psi_+}{\nabla_{\dd{\oz}}\nabla_{\dd{\varepsilon}}\psi_+
+R(\partial_{\varepsilon},\partial_{\oz})\psi_+}\\
&\qquad\qquad=2i\int_{\Sigma}dz\wedge d\oz\,\scal[g]{\nabla_{\dd{\oz}}\psi_+}{\nabla_{\dd{\varepsilon}}\psi_+}
-2i\int_{\Sigma}dz\wedge d\oz\,\scal[g]{\psi_+}{R(\partial_{\varepsilon},\partial_{\oz})\psi_+)}\\
&\qquad\qquad=2i\int_{\Sigma}dz\wedge d\oz\,\scal[g]{R\left(\dd[\varphi]{\varepsilon},\dd[\varphi]{\oz}\right)\psi_+}{\psi_+}\\
&\qquad\qquad=-2i\int_{\Sigma}dz\wedge d\oz\,\scal[g]{R(\psi_+,\psi_+)\dd[\varphi]{\oz}}{\zeta}\\
&\qquad\qquad=\int_{\Sigma}\dvol_{\Sigma}\,\lambda^{-1}\scal[g]{R(\psi_+,\psi_+)\dd[\varphi]{\oz}}{\zeta}\\
&\qquad\qquad=\int_{\Sigma}\dvol_{\Sigma}\,\scal[g]{R(\psi_{e_+},\psi_{e_+})e_{\oz}(\varphi)}{\zeta}
\end{align*}
Putting everything together, we obtain
\begin{align*}
2\nabla^2\xi+2R^{\nabla}(e_x,\xi)e_x+2R^{\nabla}(e_y,\xi)e_y
=R(\psi_{e_+},\psi_{e_+})e_{\oz}(\varphi)
=2R(\psi_{1e_+},\psi_{2e_+})e_{\oz}(\varphi)
\end{align*}
which completes the proof.
\end{proof}

We close this section by comparing $\mA_1$ with a similar model.

\begin{Def}
For $\varphi\in C^{\infty}(\Sigma,X)$ and $\psi_1\in\Gamma(S\otimes_{\bC}\varphi^*T^{\bC}X)$,
we denote by $\mA_2$ the \emph{Dirac-harmonic action functional}
\begin{align*}
\mA_2(\varphi,\psi_1):=\int_{\Sigma}\dvol_{\Sigma}\,\left(\abs{d\varphi}^2+\scal[(]{\psi_1}{\Dirac\psi_1}\right)
\end{align*}
where $\scal[(]{\cdot}{\cdot}$ is the bundle metric on $S\otimes_{\bR}\varphi^*TX\cong S\otimes_{\bC}\varphi^*T^{\bC}X$
as in (\ref{eqnH}).
\end{Def}

Note that $\mA_2$ is not Graßmann-valued. This functional and its extrema, the Dirac-harmonic maps, were introduced in
Refs.~\onlinecite{CJLW05,CJLW06}. Their main properties are studied in the same articles.

\begin{Prp}[Ref.~\onlinecite{CJLW06}]
\label{prpDiracHarmonic}
The Euler-Lagrange equations of $\mA_2$ are
\begin{align*}
\tau(\varphi)=\mR(\varphi,\psi_1)\;,\qquad\Dirac\psi_1=0
\end{align*}
where $\mR(\varphi,\psi_1)$ is a curvature term, which vanishes for
$\psi_1\in\Gamma(S^+\otimes_{\bC}\varphi^*T^{\bC}X)$.
\end{Prp}

Knowing the Euler-Lagrange equations for both $\mA_1$ and $\mA_2$, we can now consider the common case
\begin{align*}
\varphi\in C^{\infty}(\Sigma,X)\;,\qquad\psi_1\in\Gamma(S^+\otimes_{\bC}\varphi^*T^{\bC}X)\;,\qquad
\psi_2=0\;,\qquad\xi=0
\end{align*}
Note that we do not mean to restrict $\mA_2$ to sections $\psi_1\in\Gamma(S^+\otimes_{\bC}\varphi^*T^{\bC}X)$,
in which case its second term would trivially vanish. Instead, we are interested in critical points
that happen to be sections of this subbundle, in other words allowing for variations with respect to
the full spinor bundle $S$.
By Prp. \ref{prpEulerLagrange}, Prp. \ref{prpDiracHarmonic} and (\ref{eqnTwistedDiracOperatorLocal}),
we have to compare
\begin{align*}
\tau(\varphi)=0\,,\quad\nabla_{\dd{\oz}}\psi_{1+}=0\qquad\mathrm{versus}\qquad
\tau(\varphi)=0\,,\quad\nabla_{\dd{\oz}}\psi_{1e^+}=0
\end{align*}
for $\mA_1$ and $\mA_2$, respectively. The respective second conditions are not equivalent
since the conformal factor $\lambda$ is, in general, not holomorphic:
\begin{align*}
\nabla_{\partial_{\oz}}\psi_{1+}=\nabla_{\partial_{\oz}}(\lambda^{\frac{1}{4}}\psi_{1e^+})
=\lambda^{\frac{1}{4}}\nabla_{\partial_{\oz}}\psi_{1e^+}+\partial_{\oz}(\lambda^{\frac{1}{4}})\psi_{1e^+})
\end{align*}
In particular, one cannot expect an analogon of Thm. \ref{thmSupersymmetricAction}
for $\mA_1$ replaced with $\mA_2$.

\section{Holomorphic Supercurves in Supergeometry}
\label{secHolomorphicSupercurves}

In this section, we show how the tuple $(\varphi,\psi_1,\psi_2,\xi)$ from Def.
\ref{defHolomorphicSupercurveAdhoc} constitutes a morphism $\Phi$ of supermanifolds
such that the differential equations stated are equivalent to the single
condition $\dbar_J\Phi=0$. This result, which is clearly interesting in its own right,
will turn out to be important in the proof of Thm. \ref{thmSupersymmetricAction}.
In physical terminology, $\Phi$ is called a superfield, while
$\varphi$,$\psi_1$,$\psi_2$ and $\xi$ are known as its component fields
(Refs.~\onlinecite{DF99b,Fre99}).

Concerning supermanifolds, we follow the standard Berezin-Leites approach (Ref.~\onlinecite{Lei80})
describing supermanifolds in terms of ringed spaces, our main reference being Ref.~\onlinecite{Var04}.
We use complex, not real, ringed spaces which is not standard, however.
Staying in the smooth category, considering these CS supermanifolds, as they are sometimes
called (Ref.~\onlinecite{DM99}), over the usual (real) supermanifolds, does not make much
of a difference. Our convention is mainly due to convenience: The
supermanifold $\Sigma_L$, to be introduced below, then carries only one
odd supercoordinate, which makes our computations considerably easier.

\subsection{Supermanifolds and Maps with Flesh}

A supermanifold of dimension $n|m$ is a ringed space $(M,\mO_M)$ locally isomorphic to the model space
$(\bR^n,\mO_{n|m}):=\left(\bR^n,C^{\infty}(\bR^n,\bC)\otimes\bigwedge\bC^m\right)$,
where $M$ is a manifold and $\mO_M$ is a sheaf of $\bC$-superalgebras, sections of which
are called superfunctions. A morphism of supermanifolds is a morphism
$\Phi=(\varphi,\phi):(M,\mO_M)\rightarrow (X,\mO_X)$ of
$\bC$-ringed spaces, where $\varphi:M\rightarrow X$ is a smooth function and
$\phi:\mO_X\rightarrow\varphi_*\mO_M$ is a morphism of sheaves of $\bC$-superalgebras.
In particular, $\phi$ is even, i.e. parity preserving, a fact which is important for the
differential calculus.
By definition, $(M,\mO_M)$ is covered by open sets $U\subseteq M$ such that there are isomorphisms
$(\varphi,\phi):(U,\mO_M|_U)\rightarrow(\bR^n,\mO_{n|m})$ identifying the superfunctions
(locally) with Graßmann-valued functions on $\bR^n$.
A tuple $(x^1,\ldots,x^n,\theta^1,\ldots,\theta^m)$ such that $x^j$ are coordinates of $\bR^n$
and $\theta^j$ are generators of $\bigwedge\bC^m$, all identified with sections of $\mO_M|_U$,
is called a tuple of supercoordinates.
An important construction for supermanifolds is provided by complex vector bundles $E\rightarrow M$
over a manifold $M$. Denoting the sheaf of smooth sections of the Graßmann bundle
$\bigwedge E:=\bigwedge^0 E\oplus\bigwedge^1 E\oplus\cdots\oplus\bigwedge^n E$ by $\mO_E$,
$M_E:=(M,\mO_E)$ is a supermanifold, which is called split associated to $E$.
For split supermanifolds, super coordinates are provided by frames of local sections of the
vector bundle. Examples include the superpoint
$\bC^{0|m}:=\left(\{\mathrm{point}\},\bigwedge\bC^m\right)=\{\mathrm{point}\}_{\underline{\bC}^m}$
of dimension $0|m$ which is split associated to the (constant) vector bundle $\underline{\bC}^m\rightarrow\mathrm{point}$.
Moreover, any manifold $M$ may be identified with the supermanifold
$(M,\mO_{n|0}=C^{\infty}(M,\bC))$ of dimension $n|0$.

Let $(M,\mO_M)$ be a supermanifold and $U\subseteq M$. We define
$\mathrm{Der}(\mO_M(U))\subseteq\mathrm{End}_{\bC}(\mO_M(U))$
to be the complex vector space of superderivations of $\mO_M(U)$.
$\mS M:=\mathrm{Der}\mO_M$ is called the super tangent sheaf of $(M,\mO_M)$,
and sections thereof are called super vector fields.
For a smooth manifold $(M,C^{\infty}(M,\bC))$,
$\mS M=T^{\bC}M$ is the complexified tangent bundle (considered as a sheaf).
Prescribing $(f\cdot Y)(g):=f\cdot Y(g)$ makes $\mS M(U)$ a
left supermodule for the super algebra of superfunctions $\mO_M(U)$,
and setting $Y\cdot f:=(-1)^{\abs{f}{\abs{Y}}}f\cdot Y$ (for homogeneous
elements) gives $\mS M(U)$ the structure of a right $\mO_M(U)$-supermodule,
where $\abs{\cdot}\in\{0,1\}$ denotes the parity of an element in a $\bZ_2$-graded space.
As usual, we let $\dd{x^j}$ and $\dd{\theta^j}$ denote the super vector fields
on the model space $(\bR^n,\mO_{n|m})$, which are induced by the canonical (global) supercoordinates
$(x^1,\ldots,x^n,\theta^1,\ldots,\theta^m)$.
$\mathrm{Der}\mO_{n|m}(U)$ is a free $\mO_{n|m}(U)$-supermodule with $\mO_{n|m}(U)$-basis
$\left(\dd{x^1},\ldots,\dd{x^n},\dd{\theta^1},\ldots,\dd{\theta^m}\right)$.
For a general supermanifold $(M,\mO_M)$ of dimension $n|m$, $\mS M(U)$ is, therefore, free of rank $n|m$
provided that $U\subseteq M$ is contained in a super coordinate chart.

\subsubsection*{Differential Calculus and Tensors}

Let $\Phi=(\varphi,\phi):(M,\mO_M)\rightarrow(X,\mO_X)$ be a morphism of supermanifolds.
Super vector fields along $\Phi$ are sections of the sheaf (over $X$) of derivations
$\mS\Phi:=\mathrm{Der}(\mO_X,\varphi_*\mO_M)\subseteq Hom_{\bC}(\mO_X,\varphi_*\mO_M)$ along $\Phi$.
We define the differential of $\Phi$ to be the sheaf morphism
\begin{align*}
d\Phi:\varphi_*\mS M\rightarrow\mS\Phi\;,\qquad d\Phi_V(Y):=Y\circ\phi_V
\end{align*}
for $Y\in \mO_M(U)$ with $U:=\varphi^{-1}(V)$.

\begin{Lem}
\label{lemVectorFieldsAlongPhi}
Let $\Phi=(\varphi,\phi):(M,\mO_M)\rightarrow(X,\mO_X)$ be a morphism of supermanifolds and $V\subseteq X$.
Let $(Y_1,\ldots,Y_{r+s})$ be an $\mO_X|_V$-basis of $\mS X|_V$.
Then $(\phi\circ Y_1,\ldots,\phi\circ Y_{r+s})$ is a basis of $\mS\Phi|_V$, i.e. on $V$
(slightly abusing notation),
$S\Phi=\mathrm{span}_{\varphi_*\mO_M}\left(\phi\circ Y_1,\ldots,\phi\circ Y_{r+s}\right)$.
\end{Lem}

Local differential calculus works analogous to the situation of ordinary manifolds.
Let $(\xi^1,\ldots,\xi^{n+m})$ and $(\eta^1,\ldots,\eta^{r+s})$ be local super coordinates of
$(M,\mO_M)$ and $(X,\mO_X)$, respectively, where, in order not to overload notation,
we leave the domains of the super coordinates implicit and treat even and odd coordinates on an equal footing.
Let $Y\in\mS M$ be a super vector field on $(M,\mO_M)$,
and let $Y^j$ be super functions such that (locally) $Y=\dd{\xi^j}\cdot Y^j$.
Then, abbreviating $\dd[\phi(\eta^i)]{\xi^k}:=\dd{\xi^k}\left(\phi(\eta^i)\right)$, we have
\begin{align*}
d\Phi(Y)&=d\Phi\left(\dd{\xi^k}\cdot Y^k\right)=\left(\dd{\xi^k}\circ\phi\right)\cdot Y^k
=\dd[\phi(\eta^i)]{\xi^k}\cdot\left(\phi\circ\dd{\eta^i}\right)\cdot Y^k\\
&=\left(\phi\circ\dd{\eta^i}\right)
\left((-1)^{(\abs{\xi^k}+\abs{\eta^i})\cdot\abs{\eta^i}}\dd[\phi(\eta^i)]{\xi^k}\right)\cdot Y^k\\
&=:\left(\phi\circ\dd{\eta^i}\right)\cdot d\Phi^i_{\phantom{i}k}\cdot Y^k
\end{align*}
The last equation defines $d\Phi^i_{\phantom{i}k}$, which reduces to
$d\Phi^i_{\phantom{i}k}=\dd[\phi(\eta^i)]{\xi^k}$ if
the target space is an ordinary manifold (all $\eta^i$ are even).
We introduce half-index notation
\begin{align*}
d\Phi_{\phantom{i}k}:=d\Phi\left(\dd{\xi^k}\right)=\left(\phi\circ\dd{\eta^i}\right)\cdot d\Phi^i_{\phantom{i}k}
\end{align*}

Let $\Phi=(\varphi,\phi):(M,\mO_M)\rightarrow(X,\mO_X)$ be a morphism of supermanifolds
and consider the sheaf $\mS X$ of $\mO_X$-supermodules and the sheaf $\mS\Phi$ of $\varphi_*\mO_M$-supermodules.
Tensors on the super tangent sheaf can be pulled back as in the case of ordinary manifolds.

\begin{Lem}
\label{lemSuperPullback}
Let $E\in\mathrm{End}_{\mO_X}(\mS X)$ and $B\in\Hom_{\mO_X}(\mS X\otimes_{\mO_X}\mS X,\mO_X)$ be sections
of the sheaves of superlinear endomorphisms and superbilinear maps, respectively. Prescribing
\begin{align*}
E_{\Phi}(\phi\circ Y):=\phi\circ E(Y)\;,\qquad
\scal[B_{\Phi}]{\phi\circ Y}{\phi\circ Z}:=\phi\circ\scal[B]{Y}{Z}
\end{align*}
for $Y,Z\in\mS X$, together with super(bi)linear extensions for general sections of $\mS\Phi$, yields
well-defined sections $E_{\Phi}\in\mathrm{End}_{\varphi_*\mO_M}(\mS\Phi)$ and
$B_{\Phi}\in\Hom_{\varphi_*\mO_M}(\mS\Phi\otimes_{\varphi_*\mO_M}\mS\Phi,\varphi_*\mO_M)$.
\end{Lem}

\begin{proof}
A short calculations shows that this prescription does not depend on the sections of $\mS X$,
and the resulting objects satisfy all properties claimed (compare also Ref.~\onlinecite{Goe08}
for the special case of semi-Riemannian supermetrics.).
\end{proof}

\begin{Def}
\label{defAlmostComplexStructure}
An \emph{(even) almost complex structure} on a supermanifold $(M,\mO_M)$ is an even section
$J\in\mathrm{End}_{\mO_M}(\mS M)$ such that $J^2=-\id$.
\end{Def}

We consider only even structures here. Consult Sec. 5.2. in Ref.~\onlinecite{Sac07} and the
references therein for a discussion of that matter.
Def. \ref{defAlmostComplexStructure} leads to the notion of an almost complex supermanifold.
In particular, complex supermanifolds are almost complex (Ref.~\onlinecite{HW87}).
Note that, for the last conclusion, it does not suffice that the underlying smooth manifold
is complex. One can show that the split supermanifold associated to a holomorphic vector bundle
over a complex manifold is complex.
In order to avoid a thorough treatment of complex supermanifolds
and since the explicit form is hard to find in the literature,
we construct the resulting almost complex structure next.
We need the notion of holomorphic split coordinates on $M_E$,
that is supercoordinates $\Phi=(\varphi,\phi):(U\subseteq M,\mO_E|_U)\rightarrow(\bR^{2n},\mO_{2n|m})$
such that $\varphi=(x^1,y^1,\ldots,x^n,y^n):U\rightarrow\bR^{2n}$ are holomorphic coordinates and $\phi$ identifies the
Graßmann generators $\theta^j$ on $\mO_{2n|m}$ with (nonvanishing) holomorphic sections of $E|_U$.

\begin{Lem}
\label{lemAlmostComplexStructure}
Let $M_E=(M,\mO_E)$ be the split supermanifold associated to a holomorphic vector bundle $E\rightarrow M$
over a complex manifold $M$.
Then $M_E$ carries a canonical almost-complex structure $j$ as follows.
Let $\Phi=(\varphi,\phi):(U,\mO_E)\rightarrow(\bR^{2n},\mO_{2n|m})$ be holomorphic split coordinates
and (locally) define $j(Y):=d\Phi^{-1}\circ j_{n|m}\circ d\Phi\circ Y$
for $Y\in\mS M$, where $j_{n|m}$ is the standard complex structure on $(\bR^{2n},\mO_{2n|m})$ defined by
\begin{align*}
j_{n|m}\left(\dd{x^i}\right):=\dd{y^i}\;,\qquad j_{n|m}\left(\dd{y^i}\right):=-\dd{x^i}\;,\qquad
j_{n|m}\left(\dd{\theta^k}\right):=i\cdot\dd{\theta^k}
\end{align*}
\end{Lem}

\begin{proof}
The property $j^2=-\id$ is obvious, so it remains to show that the prescription is well-defined, i.e.
does not depend on $\Phi$. Let $\Phi_1$ and $\Phi_2$ be two holomorphic split charts and denote the induced
transition morphism by
\begin{align*}
\Phi_{12}:=\Phi_1\circ\Phi_2^{-1}=(\varphi,\phi):(\bR^{2n},\mO_{2n|m})\rightarrow(\bR^{2n},\mO_{2n|m})
\end{align*}
We claim that
\begin{align}
\label{eqnSuperTransition}
d\Phi_{12}\circ j_{n|m}=j_{n|m}\circ d\Phi_{12}
\end{align}
holds. Then well-definedness follows immediately by the following one-liner.
\begin{align*}
d\Phi_2^{-1}\circ j_{n|m}\circ d\Phi_2=d\Phi_1^{-1}\circ d\Phi_{12}\circ j_{n|m}\circ d\Phi_{12}^{-1}\circ d\Phi_1
=d\Phi_1^{-1}\circ j_{n|m}\circ d\Phi_1
\end{align*}
We show (\ref{eqnSuperTransition}). By definition, $\Phi_{12}=(\varphi_{12},\phi_{12})$ is composed
of a holomorphic transition map
$\varphi_{12}=(a^1,b^1,\ldots,a^n,b^n):\bR^{2n}\rightarrow\bR^{2n}$ and a sheaf morphism $\phi_{12}$ that corresponds to a holomorphic transition map
$D\in C^{\infty}(\bR^{2n},GL(m,\bC))$ of the bundle $E$ such that $\phi(\theta^k)=D^k_{\phantom{k}l}\cdot\theta^l$,
upon writing $D$ in matrix form. We calculate
\begin{align*}
d\Phi_{12}\circ j_{n|m}\left(\dd{x^i}\right)&=d\Phi_{12}\circ\dd{y^i}=\dd{y^i}\circ\phi\\
&=\left(\phi\circ\dd{a^k}\right)\dd[\phi(a^k)]{y^i}+\left(\phi\circ\dd{b^k}\right)\dd[\phi(b^k)]{y^i}
-\left(\phi\circ\dd{\theta^k}\right)\dd[\phi(\theta^k)]{y^i}\\
&=\left(\phi\circ\dd{a^k}\right)\dd[a^k]{y^i}+\left(\phi\circ\dd{b^k}\right)\dd[b^k]{y^i}
-\left(\phi\circ\dd{\theta^k}\right)\dd[D^k_{\phantom{k}l}]{y^i}\theta^l
\end{align*}
and
\begin{align*}
&j_{n|m}\circ d\Phi_{12}\left(\dd{x^i}\right)=j_{n|m}\circ\left(\dd{x^i}\circ\phi\right)\\
&\qquad\qquad=j_{n|m}\circ\left(\left(\phi\circ\dd{a^k}\right)\dd[a^k]{x^i}+\left(\phi\circ\dd{b^k}\right)\dd[b^k]{x^i}
-\left(\phi\circ\dd{\theta^k}\right)\dd[D^k_{\phantom{k}l}]{x^i}\theta^l\right)\\
&\qquad\qquad=\left(\phi\circ j_{n|m}\dd{a^k}\right)\dd[a^k]{x^i}+\left(\phi\circ j_{n|m}\dd{b^k}\right)\dd[b^k]{x^i}
-\left(\phi\circ j_{n|m}\dd{\theta^k}\right)\dd[D^k_{\phantom{k}l}]{x^i}\theta^l\\
&\qquad\qquad=\left(\phi\circ\dd{b^k}\right)\dd[a^k]{x^i}-\left(\phi\circ\dd{a^k}\right)\dd[b^k]{x^i}
-\left(\phi\circ\dd{\theta^k}\right)i\cdot\dd[D^k_{\phantom{k}l}]{x^i}\theta^l
\end{align*}
Comparing coefficients (which is feasible by Lem. \ref{lemVectorFieldsAlongPhi}), we find
that (\ref{eqnSuperTransition}) applied to $\dd{x^i}$ holds if and only if
\begin{align*}
\dd[a^k]{y^i}=-\dd[b^k]{x^i}\;,\qquad\dd[b^k]{y^i}=\dd[a^k]{x^i}\;,\qquad\dd[D^k_{\phantom{k}l}]{y^i}=i\dd[D^k_{\phantom{k}l}]{x^i}
\end{align*}
The first two equations are just the Cauchy-Riemann equations for $\varphi$
and, writing $D^k_{\phantom{k}l}$ in the form $D^k_{\phantom{k}l}=:A+iB$, the third equation is equivalent to
the Cauchy-Riemann equations $\dd[A]{y^i}=-\dd[B]{x^i}$ and $\dd[B]{y^i}=\dd[A]{x^i}$ for $D$.
(\ref{eqnSuperTransition}) applied to $\dd{\theta^i}$
is automatically satisfied while considering $\dd{y^i}$ does not lead to any new conditions.
Since $(\dd{x^i},\dd{y^i},\dd{\theta^i})$ is a (local) basis of the super tangent sheaf,
(\ref{eqnSuperTransition}) is thus established.
\end{proof}

\subsubsection*{Maps with Flesh}

Morphisms $\Phi=(\varphi,\phi):(M,\mO_M)\rightarrow(X,\mO_X)$ of supermanifolds are even.
On the other hand, one is interested in component fields of $\phi$ which are odd.
This apparent contradiction is solved by considering maps with flesh, i.e. morphisms
$(M,\mO_M)\times(B,\mO_B)\rightarrow(X,\mO_X)$ where $(B,\mO_B)$ is another supermanifold.
Consult Ref.~\onlinecite{Hel09} as well as Refs.~\onlinecite{Khe05,DF99b} for details.
In the following we fix, for simplicity, $N\in\bN$ and let $B:=\bC^{0|N}$ be the superpoint.
By construction, the supermanifold with flesh associated to a split supermanifold $M_E$, which corresponds
to a complex vector bundle $E\rightarrow M$, is split with respect to the bundle
$E\oplus\underline{\bC}^N\rightarrow M$, where $\underline{\bC}^N$ denotes the trivial bundle:
\begin{align*}
M_E\times\bC^{0|N}=(M,\mO_E)\times(\{\mathrm{point}\},\mO_{\underline{\bC}^N})\cong (M,\mO_{E\oplus\underline{\bC}^N})
\end{align*}

If $(x^1,\ldots,x^n,\theta^1,\ldots,\theta^m)$ are local supercoordinates of a supermanifold
$(M,\mO_M)$ and $\eta^1,\ldots,\eta^N$ are generators of $\bigwedge\bC^N$, then the associated
supermanifold with flesh has local supercoordinates $(x^1,\ldots,x^n,\theta^1,\ldots,\theta^m,\eta^1,\ldots,\eta^N)$.
We define its \emph{super tangent sheaf with flesh} to be the subsheaf of the original super tangent sheaf as follows.
\begin{align*}
\mS_F M:=\mathrm{span}_{\mO_M\times\bC^{0|N}}\left(\dd{x^1},\ldots,\dd{x^n},\dd{\theta^1},\ldots,\dd{\theta^m}\right)
\subseteq\mS\left((M,\mO_M)\times\bC^{0|N}\right)
\end{align*}

\begin{Def}
A \emph{map with flesh} $\Phi_F:(M,\mO_M)\rightarrow(X,\mO_X)$
is a morphism of supermanifolds $\Phi:(M,\mO_M)\times\bC^{0|N}\rightarrow(X,\mO_X)$. The differential
$d\Phi_F$ of a map with flesh is the ordinary differential restricted to $\mS_F M$:
\begin{align*}
d\Phi_F:=d\Phi|_{\varphi_*\mS_FM}:\varphi_*\mS_FM\rightarrow\mS\Phi
\end{align*}
\end{Def}

Let $E\in\mathrm{End}_{\mO_M}(\mS M)$ be an endomorphism of the (original) super tangent sheaf.
By superlinear extension over $\bigwedge\bC^N$, we obtain a unique extension to an endomorphism
on the super tangent sheaf with flesh, denoted $E_F\in\mathrm{End}_{\mO_M\otimes\bigwedge\bC^N}(\mS_F M)$.
In particular, this applies to almost complex structures $J$. Note that
the analogous construction works for general tensors on $\mS M$.
On the other hand, let $\Phi_F:(M,\mO_M)\rightarrow(X,\mO_X)$ be a map with flesh and
$E\in\mathrm{End}_{\mO_X}(\mS X)$ and $B\in\Hom_{\mO_X}(\mS X\otimes_{\mO_X}\mS X,\mO_X)$ be an
endomorphism and a bilinear form on the super tangent sheaf of the target supermanifold, respectively.
As in Lem. \ref{lemSuperPullback}, one may consider the pullback tensors $E_{\Phi}$ and $B_{\Phi}$,
which are defined by the associated morphism of supermanifolds
$\Phi:(M,\mO_M)\times\bC^{0|N}\rightarrow(X,\mO_X)$.

\begin{Def}
\label{defHolomorphicFleshMap}
Let $J^M$ and $J^X$ be almost complex structures on the supermanifolds $(M,\mO_M)$ and $(X,\mO_X)$, respectively.
Then a \emph{holomorphic map with flesh} is a map with flesh $\Phi_F:(M,\mO_M)\rightarrow(X,\mO_X)$ such that
\begin{align*}
\dbar\Phi_F:=\frac{1}{2}\left(d\Phi_F+J^X_{\Phi}\circ d\Phi_F\circ J^M_F\right)=0
\in\Hom_{\varphi_*(\mO_M\otimes\bigwedge\bC^N)}(\varphi_*S_F M,\mS\Phi)
\end{align*}
\end{Def}

Now we consider maps with flesh having an ordinary smooth manifold $X$ as target space.
In general, for every morphism $\Phi=(\varphi,\phi):(\bR^n,\mO_{n|m})\rightarrow X$,
there are unique vector fields $\xi_I\in\Gamma(\bR^n, \varphi^*T^{\bC}X)$
such that $\phi(f)=e^{\hat{\Xi}}f$ where $\hat{\Xi}$ is an extension (which is not unique)
of $\Xi:=\sum_I\xi_I\eta^I$ to a Graßmann valued vector field along $\pi:\bR^n\times X\rightarrow X$,
the $\eta^j$ are generators of $\bigwedge\bC^m$ and $I$ is some multiindex.
This follows from the proof of Thm. 1.1 in Ref.~\onlinecite{Hel08}.
If $M_E=(M,\mO_E)$ is a split supermanifold and $\Phi:M_E\rightarrow X$ is a morphism of
supermanifolds then, passing to local supercoordinates, the sections thus obtained
fit together to a global vector field $\Xi\in\Gamma(M,\varphi^*T^{\bC}X\otimes_{\bC}\bigwedge E)$,
and accordingly if we consider maps with flesh $\Phi_F:M_E\rightarrow X$.
The following result is a special case of this general principle. The direct proof
resembles a calculation in Sec. 4.1.1 of Ref.~\onlinecite{Hel09}.

\begin{Prp}
\label{prpMorphismStructure}
Let $\Sigma_L$ denote the split supermanifold associated to a complex line bundle $L\rightarrow\Sigma$
over a Riemann surface $\Sigma$ and $X$ be an ordinary manifold. Let
\begin{align*}
\Phi=(\varphi,\phi):\Sigma_L\times\bC^{0|2}\cong(M,\mO_{L\oplus\underline{\bC}^2})\rightarrow X
\end{align*}
be a morphism of supermanifolds (i.e. a map with flesh such that $N=2$). Then there are sections
$\xi\in\Gamma(\Sigma,\varphi^*T^{\bC}X)$ and $\psi_1,\psi_2\in\Gamma(\Sigma,L\otimes_{\bC}\varphi^*T^{\bC}X)$
such that
\begin{align}
\label{eqnMorphismStructure}
\phi=(\varphi^*+\eta^1\eta^2\xi)+(\eta^1\psi_1+\eta^2\psi_2)=:\phi_0+\psi
\end{align}
where $\xi$ acts on $f\in C^{\infty}(X,\bC)$ by (the complex linear extension of) $\xi(f):=df[\xi]$
and analogous for $\psi_1$ and $\psi_2$.
The correspondence $\Phi\cong(\varphi,\psi_1,\psi_2,\xi)$ is bijective.
\end{Prp}

\subsection{Holomorphic Supercurves}

In the rest of this section, we define holomorphic supercurves using the supermanifold theory
developed so far and prove equivalence of this definition with the original ad-hoc
Definition \ref{defHolomorphicSupercurveAdhoc}.
Let $L\rightarrow\Sigma$ be a holomorphic line bundle over a Riemann surface $\Sigma$
and $\Sigma_L$ be the associated split supermanifold.
By Lem. \ref{lemAlmostComplexStructure}, it carries a canonical almost complex structure $j$
which, in holomorphic split supercoordinates $(s,t,\theta)$, is given by
\begin{align}
\label{eqnAlmostComplexStructure}
j\left(\dd{s}\right)=\dd{t}\;,\qquad j\left(\dd{t}\right)=-\dd{s}\;,\qquad
j\left(\dd{\theta}\right)=i\cdot\dd{\theta}
\end{align}
On the other hand, let $(X,J)$ be an ordinary almost complex manifold.
We consider maps with flesh $\Phi_F:\Sigma_L\rightarrow X$ with respect to the superpoint with $N=2$ odd dimensions.

\begin{Def}
\label{defHolomorphicSupercurve}
A \emph{holomorphic supercurve} is a map with flesh $\Phi_F:\Sigma_L\rightarrow X$
which is holomorphic with respect to $j$ and $J$ in the sense of Def. \ref{defHolomorphicFleshMap}, i.e.
which satisfies $\dbar_J\Phi_F=0$.
\end{Def}

\begin{Lem}
\label{lemHolomorphicSupercurveLocal}
A map with flesh $\Phi_F:\Sigma_L\rightarrow X$ is a holomorphic supercurve if and only if,
upon identification with a tuple $(\varphi,\psi_1,\psi_2,\xi)$ as in Prp. \ref{prpMorphismStructure},
the following condition is satisfied. Let $(s,t,\theta)$ and $\{x^i\}$ denote
holomorphic split supercoordinates on $\Sigma_L$ and coordinates on $X$, respectively. Then
\begin{align*}
&\dd[\varphi^i]{s}+J^i_{\phantom{i}k}\cdot\dd[\varphi^k]{t}=0\;,\qquad
\psi_{\theta}^i+i\,J^i_{\phantom{i}k}\cdot\psi_{\theta}^k=0\;,\qquad
\psi_{\theta}^l\psi_{\theta}^k\dd{x^l}(J^i_{\phantom{i}k})=0\;,\\
&\dd[\psi_{\theta}^i]{s}+J^i_{\phantom{i}k}\cdot\dd[\psi_{\theta}^k]{t}
+\psi_{\theta}(J^i_{\phantom{i}k})\cdot\dd[\varphi^k]{t}=0\;,\qquad
\dd[\xi^i]{s}+J^i_{\phantom{i}k}\cdot\dd[\xi^k]{t}
+\xi(J^i_{\phantom{i}k})\cdot\dd[\varphi^k]{t}=0
\end{align*}
holds true. Here, $\psi_{1\theta},\psi_{2\theta}\in\Gamma(U,\varphi^*T^{\bC}X)$ are such that
$\psi_j=\theta\cdot\psi_{j\theta}$ and, moreover, we prescribe
$\psi_{\theta}:=\eta^1\psi_{1\theta}+\eta^2\psi_{2\theta}$ and
abbreviate $J\circ\varphi$ by $J$.
\end{Lem}

\begin{proof}
Since $(\dd{s}$,$\dd{t}$,$\dd{\theta})$ is a (local) basis of $\mS_F\Sigma$,
it is clear that $\dbar_J\Phi_F=0$ vanishes
(restricted to the domain of the coordinates) if and only if $\dbar_J\Phi_F$,
individually applied to each basis vector, vanishes.
As in the proof of Lem. \ref{lemAlmostComplexStructure},
it does moreover suffice to consider only $\dd{s}$ and $\dd{\theta}$. Using half-index notation
$\Phi_F$ is, therefore, a holomorphic supercurve if and only if
\begin{align}
\label{eqnHolomorphicSupercurveLocal}
d\Phi_{\phantom{i}s}+J_{\Phi}\circ d\Phi_{\phantom{i}t}=0\;,\qquad
d\Phi_{\phantom{i}\theta}+i\,J_{\Phi}\circ d\Phi_{\phantom{i}\theta}=0
\end{align}
holds. We calculate the second equation of (\ref{eqnHolomorphicSupercurveLocal}), using
Lem. \ref{lemSuperPullback} and Prp. \ref{prpMorphismStructure}.
\begin{align*}
d\Phi_{\phantom{i}\theta}=\left(\phi\circ\dd{x^i}\right)\cdot\dd[\phi(x^i)]{\theta}
=\left(\phi\circ\dd{x^i}\right)\cdot\psi_{\theta}^i
\end{align*}
and
\begin{align*}
J_{\Phi}\circ d\Phi_{\phantom{i}\theta}
&=J_{\Phi}\left(\phi\circ\dd{x^k}\right)\cdot\psi_{\theta}^k
=\left(\phi\circ J\left(\dd{x^k}\right)\right)\cdot\psi_{\theta}^k
=\phi\circ\left(J^i_{\phantom{i}k}\dd{x^i}\right)\cdot\psi_{\theta}^k\\
&=\phi(J^i_{\phantom{i}k})\cdot\left(\phi\circ\dd{x^i}\right)\cdot\psi_{\theta}^k
=\left(\phi\circ\dd{x^i}\right)\cdot\phi(J^i_{\phantom{i}k})\cdot\psi_{\theta}^k
\end{align*}
Comparing coefficients of the $\mS\Phi$-basis $\{\phi\circ\dd{x^i}\}_i$, we conclude that
the second equation of (\ref{eqnHolomorphicSupercurveLocal}) is equivalent to
$\psi_{\theta}^i+i\,\phi(J^i_{\phantom{i}k})\cdot\psi_{\theta}^k=0$.
Sorting into terms with and without $\theta$, we thus yield
\begin{align*}
0=\psi_{\theta}^i+i\,\phi_0(J^i_{\phantom{i}k})\cdot\psi_{\theta}^k
=\psi_{\theta}^i+i\,J^i_{\phantom{i}k}\cdot\psi_{\theta}^k\;,\qquad
0=\psi_{\theta}(J^i_{\phantom{i}k})\cdot\psi_{\theta}^k
\end{align*}
where $\phi_0:=\varphi^*+\eta^1\eta^2\xi$ and, for the second equation, we used that $(\eta^1)^2=(\eta^2)^2=0$.
In other words, we have shown that
$d\Phi_{\phantom{i}\theta}+i\,J_{\Phi}\circ d\Phi_{\phantom{i}\theta}=0$ holds if and only if
\begin{align*}
\psi_{\theta}^i+i\,J^i_{\phantom{i}k}\cdot\psi_{\theta}^k=0\;,\qquad
\psi_{\theta}^l\psi_{\theta}^k\dd{x^l}(J^i_{\phantom{i}k})=0
\end{align*}
Similarly, we calculate $d\Phi_{\phantom{i}s}=\left(\phi\circ\dd{x^i}\right)\cdot\dd[\phi(x^i)]{s}$ and
\begin{align*}
J_{\Phi}\circ d\Phi_{\phantom{i}t}&=J_{\Phi}\left(\phi\circ\dd{x^k}\right)\cdot\dd[\phi(x^k)]{t}
=\phi\circ\left(J^i_{\phantom{i}k}\dd{x^i}\right)\cdot\dd[\phi(x^k)]{t}\\
&=\left(\phi\circ\dd{x^i}\right)\cdot\phi(J^i_{\phantom{i}k})\cdot\dd[\phi(x^k)]{t}
\end{align*}
The first equation of (\ref{eqnHolomorphicSupercurveLocal}) is thus
equivalent to $\dd[\phi(x^i)]{s}+\phi(J^i_{\phantom{i}k})\cdot\dd[\phi(x^k)]{t}=0$.
Setting $\phi_0^i:=\phi_0(x^i)$ and sorting into terms without and with $\theta$, we yield
\begin{align*}
\dd[\phi_0^i]{s}+\phi_0(J^i_{\phantom{i}k})\cdot\dd[\phi_0^k]{t}=0\;,\qquad
\dd[\psi_{\theta}^i]{s}+J^i_{\phantom{i}k}\cdot\dd[\psi_{\theta}^k]{t}
+\psi_{\theta}(J^i_{\phantom{i}k})\cdot\dd[\phi_0^k]{t}=0
\end{align*}
Further sorting into terms with and without $\eta$-terms, we see that
$d\Phi_{\phantom{i}s}+\,J_{\Phi}\circ d\Phi_{\phantom{i}t}=0$ holds if and only if
\begin{align*}
0&=\dd[\varphi^i]{s}+J^i_{\phantom{i}k}\cdot\dd[\varphi^k]{t}\;,\qquad
0=\dd[\xi^i]{s}+J^i_{\phantom{i}k}\cdot\dd[\xi^k]{t}
+\xi(J^i_{\phantom{i}k})\cdot\dd[\varphi^k]{t}\\
0&=\dd[\psi_{\theta}^i]{s}+J^i_{\phantom{i}k}\cdot\dd[\psi_{\theta}^k]{t}
+\psi_{\theta}(J^i_{\phantom{i}k})\cdot\dd[\varphi^k]{t}
\end{align*}
is satisfied, which concludes the proof.
\end{proof}

We shall bring the local conditions from Lem. \ref{lemHolomorphicSupercurveLocal} into a more concise form next.

\begin{Lem}
\label{lemNijenhuisPsi}
Let $\Phi_F:\Sigma_L\rightarrow X$ be a map with flesh, $(s,t,\theta)$ and $\{x^i\}$ be as in
Lem. \ref{lemHolomorphicSupercurveLocal} and assume that
$\psi_{\theta}^i+i\,J^i_{\phantom{i}k}\cdot\psi_{\theta}^k=0$ holds. Then
\begin{align*}
\psi_{\theta}^l\psi_{\theta}^k\dd{x^l}(J^i_{\phantom{i}k})
=-\frac{i}{2}\eta^1\eta^2 N_J^{\bC}(\psi_{1\theta},\psi_{2\theta})
\end{align*}
where $N_J^{\bC}$ denotes the complex linear extension of the Nijenhuis tensor.
\end{Lem}

\begin{proof}
In the local frame $\{\partial_i:=\dd{x^i}\}$, the Nijenhuis tensor reads
\begin{align*}
N_{ij}^p=J^k_{\phantom{k}i}\partial_k(J^p_{\phantom{p}j})-J^m_{\phantom{m}j}\partial_m(J^p_{\phantom{p}i})
+J^p_{\phantom{p}k}\partial_j(J^k_{\phantom{k}i})-J^p_{\phantom{p}m}\partial_i(J^m_{\phantom{m}j})
\end{align*}
such that $N_{ij}^p\partial_p=N(\partial_i,\partial_j)$.
For the following calculation, we extend $N_J^{\bC}$ to a $\bigwedge\bC^2$-bilinear tensor as in
the discussion in the beginning of Sec. \ref{secSusySigmaModels}.
We use that, by assumption, $J^a_{\phantom{a}b}\psi_{\theta}^b=i\cdot\psi_{\theta}^a$ holds
and $\psi^a_{\theta}\cdot\psi^b_{\theta}=-\psi^b_{\theta}\cdot\psi^a_{\theta}$, to calculate
\begin{align*}
N_J^{\bC}(\psi_{\theta},\psi_{\theta})^p
&=\left(J^k_{\phantom{k}i}\partial_k(J^p_{\phantom{p}j})-J^m_{\phantom{m}j}\partial_m(J^p_{\phantom{p}i})
+J^p_{\phantom{p}k}\partial_j(J^k_{\phantom{k}i})-J^p_{\phantom{p}m}\partial_i(J^m_{\phantom{m}j})\right)\psi_{\theta}^i\psi_{\theta}^j\\
&=i\psi_{\theta}^k\psi_{\theta}^j\partial_k(J^p_{\phantom{p}j})-\psi_{\theta}^ii\psi_{\theta}^k\partial_k(J^p_{\phantom{p}i})
+J^p_{\phantom{p}k}\left(\psi_{\theta}^i\psi_{\theta}^j\partial_j(J^k_{\phantom{k}i})
-\psi_{\theta}^i\psi_{\theta}^j\partial_i(J^k_{\phantom{k}j})\right)\\
&=2i\psi_{\theta}^k\psi_{\theta}^j\partial_k(J^p_{\phantom{p}j})+2J^p_{\phantom{p}k}\psi_{\theta}^i\psi_{\theta}^j\partial_j(J^k_{\phantom{k}i})\\
&=2i\psi_{\theta}^k\psi_{\theta}^j\partial_k(J^p_{\phantom{p}j})-2\partial_j(J^p_{\phantom{p}k})J^k_{\phantom{k}i}\psi_{\theta}^i\psi_{\theta}^j\\
&=2i\psi_{\theta}^k\psi_{\theta}^j\partial_k(J^p_{\phantom{p}j})-2i\psi_{\theta}^k\psi_{\theta}^j\partial_j(J^p_{\phantom{p}k})\\
&=4i\psi_{\theta}^k\psi_{\theta}^j\partial_k(J^p_{\phantom{p}j})
\end{align*}
We thus obtain
\begin{align*}
\psi_{\theta}^l\psi_{\theta}^k\dd{x^l}(J^i_{\phantom{i}k})=-\frac{i}{4}N_J^{\bC}(\psi_{\theta},\psi_{\theta})^i
=-\frac{i}{4}\eta^1\eta^2\left(N_J^{\bC}(\psi_{1\theta},\psi_{2\theta})^i-N_J^{\bC}(\psi_{2\theta},\psi_{1\theta})^i\right)
\end{align*}
and the statement follows from skew-symmetry of $N_J$.
\end{proof}

\begin{Lem}
\label{lemDXiLocal}
Let $\varphi\in C^{\infty}(\Sigma,X)$ and $\xi\in\Gamma(\Sigma,\varphi^*T^{\bC}X)$,
and assume that $\dbar_J\varphi=0$. Let $(s,t)$ and $\{x^i\}$ be holomorphic coordinates on $\Sigma$ and
coordinates on $X$, respectively. Then
\begin{align*}
\dd[\xi^i]{s}+J^i_{\phantom{i}k}\cdot\dd[\xi^k]{t}
+\xi(J^i_{\phantom{i}k})\cdot\dd[\varphi^k]{t}=0
\end{align*}
holds if and only if $D^{\bC}_{\varphi}\xi=0$ (restricted to the domain of the coordinates),
where $D^{\bC}_{\varphi}$ denotes the complex linear extension of the linearised $\dbar_J$-operator
(\ref{eqnLinearisedDbar}).
\end{Lem}

\begin{proof}
A short calculation in coordinates (as in Ref.~\onlinecite{MS04}) yields
\begin{align*}
D_{\varphi}\xi&=\dbar_J\xi-\frac{1}{2}(J\partial_{\xi}J)(\varphi)\partial_J\varphi\\
\dbar_J\xi&=\frac{1}{2}\dd{x^i}\left(\dd[\xi^i]{s}+J^i_{\phantom{i}k}\dd[\xi^k]{t}\right)ds
+\frac{1}{2}\dd{x^i}\left(\dd[\xi^i]{t}-J^i_{\phantom{i}k}\dd[\xi^k]{s}\right)dt
\end{align*}
Moreover, using $d\varphi=\partial_J\varphi+\dbar_J\varphi=\partial_J\varphi$, we calculate
\begin{align*}
-\frac{1}{2}(J\partial_{\xi}J)(\varphi)\partial_J\varphi
&=\frac{1}{2}\partial_{\xi}J\circ J\circ d\varphi
=\frac{1}{2}\xi(J)\circ d\varphi\circ j\\
&=\frac{1}{2}\xi(J)\circ\left(\dd[\varphi]{s}ds+\dd[\varphi]{t}dt\right)\circ j\\
&=\frac{1}{2}\xi(J)\circ\left(-\dd[\varphi]{s}dt+\dd[\varphi]{t}ds\right)\\
&=\frac{1}{2}\dd{x^i}\cdot\xi(J^i_{\phantom{i}k})\left(\dd[\varphi^k]{t}ds-\dd[\varphi^k]{s}dt\right)
\end{align*}
and, therefore,
\begin{align*}
D_{\varphi}\xi&=\frac{1}{2}\dd{x^i}\left(\dd[\xi^i]{s}+J^i_{\phantom{i}k}\dd[\xi^k]{t}
+\xi(J^i_{\phantom{i}k})\cdot\dd[\varphi^k]{t}\right)ds\\
&\qquad+\frac{1}{2}\dd{x^i}\left(\dd[\xi^i]{t}-J^i_{\phantom{i}k}\dd[\xi^k]{s}
-\xi(J^i_{\phantom{i}k})\cdot\dd[\varphi^k]{s}\right)dt
\end{align*}
where the sum in the second pair of parentheses equals $J$ times the first sum.
\end{proof}

We close this section with the important result that our two definitions of holomorphic
supercurves are equivalent.

\begin{Prp}
\label{prpHolomorphicSupercurveLocal}
Let $\Phi_F:\Sigma_L\rightarrow X$ be a map with flesh, identified with a tuple
$(\varphi,\psi_1,\psi_2,\xi)$ as in Prp. \ref{prpMorphismStructure}.
Then $\Phi_F$ is a holomorphic supercurve in the sense of Def. \ref{defHolomorphicSupercurve}
if and only if $(\varphi,\psi_1,\psi_2,\xi)$
is a holomorphic supercurve in the sense of Def. \ref{defHolomorphicSupercurveAdhoc}.
\end{Prp}

\begin{proof}
This follows immediately from lemmas \ref{lemHolomorphicSupercurveLocal},
\ref{lemNijenhuisPsi} and \ref{lemDXiLocal}, the latter of which holds
verbatim with $\xi$ replaced by $\psi_{j\theta}$.
\end{proof}

\section{The Super Action Identity}
\label{secSuperActionIdentity}

In this section, we prove Thm. \ref{thmSupersymmetricAction} by means of a generalisation
of the action identity (\ref{eqnEnergyIdentity}).
For that purpose, we will first express the action functional $\mA_1$ from Def.
\ref{defSupersymmetricAction} in terms of supergeometry.
Let $(X,g)$ be a Riemannian manifold and $\Sigma$ be a closed Riemann surface with a spin structure
and spinor bundles $S^+$ and $S=S^+\oplus S^-$. As in the previous section, we denote the corresponding
split supermanifolds by $\Sigma_{S^+}=(\Sigma,\mO_{S^+})$ and $\Sigma_S=(\Sigma,\mO_S)$,
respectively. There is a canonical morphism $I=(\id,\iota):\Sigma_S\rightarrow\Sigma_{S^+}$
of supermanifolds, consisting of the identity
map $\id:\Sigma\rightarrow\Sigma$ and the embedding $\iota:\mO_{S^+}\rightarrow\mO_S$.
By concatenation with $I$, we may consider every morphism (with flesh) $\Sigma_{S^+}\rightarrow X$
implicitly as a morphism (with flesh) $\Sigma_S\rightarrow X$.
The use of the additional odd dimension will become clear during the proof of Lem.
\ref{lemNonsuperAction} below. Let $z=s+it$ be holomorphic coordinates on $\Sigma$ and $\theta^{\pm}$
be as in Def. \ref{defSquareRootCoordinates}.
The tuple $(z,\theta^+,\theta^-)$ constitutes holomorphic split
supercoordinates for $\Sigma_S$, and we extend the almost complex structure $j$ from
(\ref{eqnAlmostComplexStructure}) (with $\theta=\theta^+$) by prescribing
$j\left(\dd{\theta^-}\right):=-i\cdot\dd{\theta^-}$.
Since the transition functions of $S^+$ and $S^-$ do not mix, well-definedness follows
analogous to Lem. \ref{lemAlmostComplexStructure}.
We introduce super vector fields
\begin{align}
\label{eqnSupervectorFields}
D_+:=\dd{\theta^+}+\theta^+\dd{z}\;,\qquad D_-:=\dd{\theta^-}+\theta^-\dd{\oz}
\end{align}
on $\Sigma_S$. Here, $D_+$ may be interpreted as (the local form of) the structure of a super Riemann surface
on $\Sigma$ (Ref.~\onlinecite{LR88}), and $D_-$ as its complex conjugate.
In the following, maps with flesh are always meant with respect to the superpoint with
$N=2$ odd dimensions.

\begin{Def}
\label{defSupersymmetricLagrangian}
Let $\Phi_F:\Sigma_{S^+}\rightarrow X$ be a map with flesh.
We define the (Graßmann valued) $2$-form
\begin{align*}
\mL(\Phi):=-i\,dz\wedge d\oz\,\dd{\theta^+}\dd{\theta^-}
\,\scal[g_{\Phi}]{d\Phi(D_+)}{d\Phi(D_-)}\in\Omega^2(\Sigma,\bC)\otimes_{\bC}\bigwedge\bC^2_{\mathrm{even}}
\end{align*}
where $g_{\Phi}$ denotes the pullback of $g$ under the corresponding (ordinary) morphism
$\Phi:\Sigma_S\times\bC^{0|2}\rightarrow X$ of supermanifolds as in Lem. \ref{lemSuperPullback}.
\end{Def}

By a simple calculation, $\mL$ is independent of the supercoordinates
of the type considered. It can be expressed in a global fashion, involving
the super integral form (Refs.~\onlinecite{DM99,Sha88}) induced by the semi-Riemannian
supermetric (Ref.~\onlinecite{Goe08}) which, in turn, is induced by $h$ and a skew-symmetric
bilinear form $\eta$ on $S$ that is suitable in the sense of
Ref.~\onlinecite{ACDS97} (consult also Refs.~\onlinecite{Har90,Var04} for a classification of such forms).
By the supermetric just mentioned, together with the metric $g$ on the target manifold $X$,
one obtains, by a supertrace construction, a superfunction on $\Sigma$ which, multiplied
with the super integral form, constitues $\mL$.
This is similar to the definition of $\abs{d\varphi}^2_{h,g}\dvol_{\Sigma}$
for a map $\varphi:\Sigma\rightarrow X$ (Refs.~\onlinecite{DF99b,Fre99}).

\begin{Lem}
\label{lemNonsuperAction}
Upon identifying $\Phi$ with a tuple $(\varphi,\psi_1,\psi_2,\xi)$ as in Prp. \ref{prpMorphismStructure},
the integral over $\mL$ coincides with the action funcional $\mA_1$ from Def. \ref{defSupersymmetricAction}:
\begin{align*}
\int_{\Sigma}\mL(\Phi)=\mA_1(\varphi,\psi_1,\psi_2,\xi)
\end{align*}
\end{Lem}

\begin{proof}
Let $\{x^j\}$ be local coordinates on $X$. Using the decomposition formula (\ref{eqnMorphismStructure})
for $\phi$, we calculate
\begin{align*}
&\dd{\theta^+}\dd{\theta^-}
\,\scal[g_{\Phi}]{d\Phi\left(\dd{\theta^+}+\theta^+\dd{z}\right)}{d\Phi\left(\dd{\theta^-}+\theta^-\dd{\oz}\right)}\\
&\;\;=\dd{\theta^+}\dd{\theta^-}
\,\scal[g_{\Phi}]{\left(\phi\circ\dd{x^i}\right)\cdot\left(\psi_+^i+\theta^+\dd[\phi_0^i]{z}\right)}
{\left(\phi\circ\dd{x^j}\right)\cdot\left(\theta^-\dd[\phi_0^j]{\oz}-\theta^+\theta^-\dd[\psi_+^j]{\oz}\right)}\\
&\;\;=\dd{\theta^+}\dd{\theta^-}\,\left((\phi\circ g_{ij})\cdot\left(\psi_+^i+\theta^+\dd[\phi_0^i]{z}\right)\cdot
\left(\theta^-\dd[\phi_0^j]{\oz}-\theta^+\theta^-\dd[\psi_+^j]{\oz}\right)\right)\\
&\;\;=\dd{\theta^+}\dd{\theta^-}\,\left(\left(\phi_0(g_{ij})+\theta^+\psi_+(g_{ij})\right)
\cdot\left(\psi_+^i+\theta^+\dd[\phi_0^i]{z}\right)\cdot
\left(\theta^-\dd[\phi_0^j]{\oz}-\theta^+\theta^-\dd[\psi_+^j]{\oz}\right)\right)\\
&\;\;=\phi_0(g_{ij})\psi_+^i\dd[\psi_+^j]{\oz}
-\phi_0(g_{ij})\dd[\phi_0^i]{z}\dd[\phi_0^j]{\oz}-\psi_+(g_{ij})\psi_+^i\dd[\phi_0^j]{\oz}\\
&\;\;=-\phi_0(g_{ij})\dd[\phi_0^i]{z}\dd[\phi_0^j]{\oz}+\psi_+^i\dd[\psi_+^j]{\oz}g_{ij}
+\psi_+^i\psi_+^m\dd[\varphi^j]{\oz}\dd[g_{ij}]{x^m}
\end{align*}
where, in the last equation, we used that the $\eta^1\eta^2$-term in $\phi_0$ cancels with $\psi_+$,
abbreviating $\varphi^*g_{ij}$ by $g_{ij}$.
By the symmetry properties of the metric and the product of two odd quantities, we obtain
$\psi_+^i\psi_+^m\partial_lg_{mi}=0$ and $\psi_+^i\psi_+^m\partial_mg_{li}=-\psi_+^i\psi_+^m\partial_ig_{lm}$
and, therefore,
\begin{align*}
\scal[g]{\psi_+}{\nabla_{\dd{\oz}}\psi_+}&=g_{ij}\psi_+^i\left(\nabla_{\dd{\oz}}\psi_+\right)^j
=g_{ij}\psi_+^i\left(\dd[\psi_+^j]{\oz}+\psi_+^m\dd[\varphi^l]{\oz}\Gamma^j_{lm}\right)\\
&=\psi_+^i\dd[\psi_+^j]{\oz}g_{ij}+\psi_+^i\psi_+^m\dd[\varphi^l]{\oz}g_{ij}
\left(\frac{1}{2}g^{ja}\left(\partial_lg_{ma}+\partial_mg_{la}-\partial_ag_{lm}\right)\right)\\
&=\psi_+^i\dd[\psi_+^j]{\oz}g_{ij}+\frac{1}{2}\psi_+^i\psi_+^m\dd[\varphi^l]{\oz}
(\partial_lg_{mi}+\partial_mg_{li}-\partial_ig_{lm})\\
&=\psi_+^i\dd[\psi_+^j]{\oz}g_{ij}+\psi_+^i\psi_+^m\dd[\varphi^j]{\oz}\dd[g_{ij}]{x^m}
\end{align*}
Together with the first calculation, we thus yield
\begin{align*}
\mL(\Phi)=i\,dz\wedge d\oz\left(\phi_0(g_{ij})\dd[\phi_0^i]{z}\dd[\phi_0^j]{\oz}
-\scal[g]{\psi_+}{\nabla_{\dd{\oz}}\psi_+}\right)
\end{align*}
We further calculate
\begin{align*}
\phi_0(g_{ij})\dd[\phi_0^i]{z}\dd[\phi_0^j]{\oz}
&=(\varphi+\eta^1\eta^2\xi)(g_{ij})\dd{z}(\varphi+\eta^1\eta^2\xi)^i\dd{\oz}(\varphi+\eta^1\eta^2\xi)^j\\
&=g_{ij}\dd[\varphi^i]{z}\dd[\varphi^j]{\oz}+\eta^1\eta^2\left(\xi(g_{ij})\dd[\varphi^i]{z}\dd[\varphi^j]{\oz}
+g_{ij}\dd[\xi^i]{z}\dd[\varphi^j]{\oz}+g_{ij}\dd[\varphi^i]{z}\dd[\xi^j]{\oz}\right)
\end{align*}
and obtain the classical energy density
$i\,dz\wedge d\oz\,g_{ij}\dd[\varphi^i]{z}\dd[\varphi^j]{\oz}=\frac{1}{2}\dvol_{\Sigma}\abs{d\varphi}^2$
for the zero order term in $\mL(\Phi)$.
To proceed, consider the (global) $1$-form $\Lambda:=\scal[g]{\dd[\varphi]{\oz}}{\xi}d\oz-\scal[g]{\dd[\varphi]{z}}{\xi}dz$
on $\Sigma$. Its differential can be expressed as follows.
\begin{align*}
d\Lambda&=\dd{z}\left(\scal[g]{\dd[\varphi]{\oz}}{\xi}\right)dz\wedge d\oz+
\dd{\oz}\left(\scal[g]{\dd[\varphi]{z}}{\xi}\right)dz\wedge d\oz\\
&=\left(\dd[g_{ij}]{z}\dd[\varphi^i]{\oz}\xi^j+g_{ij}\frac{\partial^2\varphi^i}{\partial z\partial\oz}\xi^j
+g_{ij}\dd[\varphi^i]{\oz}\dd[\xi^j]{z}\right.\\
&\qquad\left.+\dd[g_{ij}]{\oz}\dd[\varphi^i]{z}\xi^j+g_{ij}\frac{\partial^2\varphi^i}{\partial\oz\partial z}\xi^j
+g_{ij}\dd[\varphi^i]{z}\dd[\xi^j]{\oz}\right)dz\wedge d\oz\\
&=\left(g_{ij}\dd[\xi^i]{z}\dd[\varphi^j]{\oz}+g_{ij}\dd[\varphi^i]{z}\dd[\xi^j]{\oz}
+\xi^j\left(\dd[g_{ij}]{\oz}\dd[\varphi^i]{z}+\dd[g_{ij}]{z}\dd[\varphi^i]{\oz}
+2g_{ij}\frac{\partial^2\varphi^i}{\partial z\partial\oz}\right)\right)dz\wedge d\oz
\end{align*}
Since $\Sigma$ is, by assumption, closed, we obtain $\int_{\Sigma}d\Lambda=0$ and hence
\begin{align*}
&\int dz\wedge d\oz\left(\xi(g_{ij})\dd[\varphi^i]{z}\dd[\varphi^j]{\oz}
+g_{ij}\dd[\xi^i]{z}\dd[\varphi^j]{\oz}+g_{ij}\dd[\varphi^i]{z}\dd[\xi^j]{\oz}\right)\\
&\qquad=\int dz\wedge d\oz\,\xi^k\left(\partial_k(g_{ij})\dd[\varphi^i]{z}\dd[\varphi^j]{\oz}
-\dd[g_{ik}]{\oz}\dd[\varphi^i]{z}-\dd[g_{ik}]{z}\dd[\varphi^i]{\oz}-2g_{ik}\frac{\partial^2\varphi^i}{\partial z\partial\oz}\right)\\
&\qquad=\int dz\wedge d\oz\,\xi^k\left(\partial_k(g_{ij})\dd[\varphi^i]{z}\dd[\varphi^j]{\oz}
-\dd[\varphi^m]{\oz}\partial_m(g_{ik})\dd[\varphi^i]{z}-\dd[\varphi^m]{z}\partial_m(g_{ik})\dd[\varphi^i]{\oz}\right.\\
&\qquad\qquad\qquad\qquad\qquad\left.-2g_{ik}\frac{\partial^2\varphi^i}{\partial z\partial\oz}\right)\\
&\qquad=\int dz\wedge d\oz\,\xi^k\left(\dd[\varphi^i]{z}\dd[\varphi^j]{\oz}
(\partial_k(g_{ij})-\partial_j(g_{ik})-\partial_i(g_{jk}))-2g_{ik}\frac{\partial^2\varphi^i}{\partial z\partial\oz}\right)\\
&\qquad=\int dz\wedge d\oz\,\xi^k\left(\dd[\varphi^i]{z}\dd[\varphi^j]{\oz}(-2g_{km}\Gamma^m_{ij})
-2g_{ik}\frac{\partial^2\varphi^i}{\partial z\partial\oz}\right)\\
&\qquad=-2\int dz\wedge d\oz\,g_{kl}\xi^k\left(\frac{\partial^2\varphi^l}{\partial z\partial\oz}
+\Gamma^l_{ij}\dd[\varphi^i]{z}\dd[\varphi^j]{\oz}\right)\\
&\qquad=-\frac{\lambda}{2}\int dz\wedge d\oz\,\scal[g]{\xi}{\tau(\varphi)}\\
&\qquad=i\int_{\Sigma}\dvol_{\Sigma}\,\scal[g]{\xi}{\tau(\varphi)}
\end{align*}
where we used the local formula for $\tau(\varphi)$ from (\ref{eqnLocalTension}).
This transformation concludes the proof of the lemma.
\end{proof}

Now let $(X,\omega)$ be a symplectic manifold, $J$ be an $\omega$-compatible almost complex structure
and $g$ be the induced Riemann metric.
For the next observation, note that the pullback tensors $\omega_{\Phi}$, $g_{\Phi}$
and $J_{\Phi}$ are, by definition, related to each other in the analogous way.
In particular, $J_{\Phi}$ is $\omega_{\Phi}$-compatible and $g_{\Phi}$ is
$J_{\Phi}$-orthogonal. As usual, we see that both terms occurring are globally well-defined.

\begin{Lem}
\label{lemLagrangianDecomposition}
If $(X,\omega)$ is a symplectic manifold with $\omega$-compatible $J$ and induced $g$,
then $\mL(\Phi)$ permits the following sum decomposition.
\begin{align*}
\mL(\Phi)&=-i\,dz\wedge d\oz\,\dd{\theta^+}\dd{\theta^-}\,
\left(2\scal[g_{\Phi}]{\dbar_J\Phi(D_+)}{\dbar_J\Phi(D_-)}
-i\scal[\omega_{\Phi}]{d\Phi(D_+)}{d\Phi(D_-)}\right)\\
&=:\mL_{\dbar_J}(\Phi)+\mL_{\omega}(\Phi)
\end{align*}
\end{Lem}

\begin{proof}
The decomposition is shown by the following straightforward calculation, using $j(D_+)=iD_+$ and $j(D_-)=-iD_-$.
\begin{align*}
L:&=\scal[g_{\Phi}]{d\Phi\circ D_+}{d\Phi\circ D_-}\\
&=\scal[g_{\Phi}]{(d\Phi+J_{\Phi}\circ d\Phi\circ j)\circ D_+}{(d\Phi+J_{\Phi}\circ d\Phi\circ j)\circ D_-}\\
&\qquad-\scal[g_{\Phi}]{d\Phi\circ D_+}{J_{\Phi}\circ d\Phi\circ j\circ D_-}
-\scal[g_{\Phi}]{J_{\Phi}\circ d\Phi\circ j\circ D_+}{d\Phi\circ D_-}\\
&\qquad-\scal[g_{\Phi}]{J_{\Phi}\circ d\Phi\circ j\circ D_+}{J_{\Phi}\circ d\Phi\circ j\circ D_-}\\
&=4\cdot\scal[g_{\Phi}]{\dbar_J\Phi\circ D_+}{\dbar_J\Phi\circ D_-}
+i\cdot\scal[g_{\Phi}]{d\Phi\circ D_+}{J_{\Phi}\circ d\Phi\circ D_-}\\
&\qquad-i\cdot\scal[g_{\Phi}]{J\circ d\Phi\circ D_+}{d\Phi\circ D_-}
-\scal[g_{\Phi}]{J_{\Phi}\circ d\Phi\circ D_+}{J_{\Phi}\circ d\Phi\circ D_-}\\
&=4\cdot\scal[g_{\Phi}]{\dbar_J\Phi\circ D_+}{\dbar_J\Phi\circ D_-}-L
-2i\cdot\scal[\omega_{\Phi}]{d\Phi\circ D_+}{d\Phi\circ D_-}
\end{align*}
\end{proof}

\begin{Prp}[Super Action Identity]
\label{prpSuperActionIdentity}
If $(X,\omega)$ is a symplectic manifold with $\omega$-compatible $J$ and induced $g$,
then the action functional permits the following sum decomposition.
\begin{align*}
\int_{\Sigma}\mL(\Phi)=\int_{\Sigma}\varphi^*\omega-2i\int_{\Sigma}dz\wedge d\oz\,
\dd{\theta^+}\dd{\theta^-}\,\scal[g_{\Phi}]{\dbar_J\Phi(D_+)}{\dbar_J\Phi(D_-)}
\end{align*}
\end{Prp}

\begin{proof}
By Lem. \ref{lemLagrangianDecomposition}, it remains to show
\begin{align}
\label{eqnLw}
\int_{\Sigma}\mL_{\omega}(\Phi)=\int_{\Sigma}\varphi^*\omega
\end{align}
By a calculation verbatim to the first part of the proof of Lem. \ref{lemNonsuperAction}
with $g$ and $g_{ij}$ replaced by $\omega$ and $\omega_{ij}$, respectively, we find:
\begin{align*}
\int_{\Sigma}\mL_{\omega}(\Phi)
=-\int_{\Sigma}dz\wedge d\oz\,\left(-\phi_0(\omega_{ij})\dd[\phi_0^i]{z}\dd[\phi_0^j]{\oz}
+\psi_+^i\dd[\psi_+^j]{\oz}\omega_{ij}
+\psi_+^i\psi_+^m\dd[\varphi^j]{\oz}\dd[\omega_{ij}]{x^m}\right)
\end{align*}
In the following, we use the closedness condition $d\omega=0$ in the form
$\dd[\omega_{ij}]{x^k}+\dd[\omega_{jk}]{x^i}+\dd[\omega_{ki}]{x^j}=0$ to yield
\begin{align*}
\dd{\oz}\scal[\omega]{\psi_+}{\psi_+}&=\dd{\oz}(\psi_+^i\psi_+^j\omega_{ij})
=\dd[\psi_+^i]{\oz}\psi_+^j\omega_{ij}+\psi_+^i\dd[\psi_+^j]{\oz}\omega_{ij}+\psi_+^i\psi_+^j\dd[\omega_{ij}]{\oz}\\
&=2\psi_+^i\dd[\psi_+^j]{\oz}\omega_{ij}+\psi_+^i\psi_+^j\dd[\omega_{ij}]{x^k}\dd[\varphi^k]{\oz}\\
&=2\psi_+^i\dd[\psi_+^j]{\oz}\omega_{ij}-\psi_+^i\psi_+^j\dd[\omega_{jk}]{x^i}\dd[\varphi^k]{\oz}
-\psi_+^i\psi_+^j\dd[\omega_{ki}]{x^j}\dd[\varphi^k]{\oz}\\
&=2\psi_+^i\dd[\psi_+^j]{\oz}\omega_{ij}+\psi_+^j\dd[\varphi^k]{\oz}\psi_+^i\dd[\omega_{jk}]{x^i}
+\psi_+^i\dd[\varphi^k]{\oz}\psi_+^j\dd[\omega_{ik}]{x^j}\\
&=2\psi_+^i\dd[\psi_+^j]{\oz}\omega_{ij}+2\psi_+^i\dd[\varphi^j]{\oz}\psi_+^k\dd[\omega_{ij}]{x^k}
\end{align*}
Hence, defining $\Omega:=\scal[\omega]{\psi_+}{\psi_+}dz$ we obtain
\begin{align*}
0=\int_{\Sigma}d\Omega&=\int_{\Sigma}d(\scal[\omega]{\psi_+}{\psi_+})\wedge dz
=-\int_{\Sigma}dz\wedge d\oz\,\dd{\oz}\scal[\omega]{\psi_+}{\psi_+}\\
&=-2\int_{\Sigma}dz\wedge d\oz\,\left(\psi_+^i\dd[\psi_+^j]{\oz}\omega_{ij}
+\psi_+^i\psi_+^m\dd[\varphi^j]{\oz}\dd[\omega_{ij}]{x^m}\right)
\end{align*}
and thus only the term
\begin{align*}
\int_{\Sigma}\mL_{\omega}(\Phi)
=\int_{\Sigma}dz\wedge d\oz\,\phi_0(\omega_{ij})\dd[\phi_0^i]{z}\dd[\phi_0^j]{\oz}
\end{align*}
remains. By a calculation verbatim as in the proof of Lem. \ref{lemNonsuperAction}
with $g$ and $g_{ij}$ replaced by $\omega$ and $\omega_{ij}$, respectively, we further calculate
\begin{align*}
\phi_0(\omega_{ij})\dd[\phi_0^i]{z}\dd[\phi_0^j]{\oz}
=\omega_{ij}\dd[\varphi^i]{z}\dd[\varphi^j]{\oz}
+\eta^1\eta^2\left(\xi(\omega_{ij})\dd[\varphi^i]{z}\dd[\varphi^j]{\oz}
+\omega_{ij}\dd[\xi^i]{z}\dd[\varphi^j]{\oz}+\omega_{ij}\dd[\varphi^i]{z}\dd[\xi^j]{\oz}\right)
\end{align*}
Now, setting $\Lambda:=-\scal[\omega]{\dd[\varphi]{z}}{\xi}dz-\scal[\omega]{\dd[\varphi]{\oz}}{\xi}d\oz$,
closedness of $\omega$ and $\Sigma$ implies
\begin{align*}
0=\int_{\Sigma}d\Lambda
=\int_{\Sigma}dz\wedge d\oz\,\left(\xi(\omega_{ij})\dd[\varphi^i]{z}\dd[\varphi^j]{\oz}
+\omega_{ij}\dd[\xi^i]{z}\dd[\varphi^j]{\oz}+\omega_{ij}\dd[\varphi^i]{z}\dd[\xi^j]{\oz}\right)
\end{align*}
Therefore, only the zero order term
\begin{align*}
\int_{\Sigma}\mL_{\omega}(\Phi)
=\int_{\Sigma}dz\wedge d\oz\,\omega_{ij}\dd[\varphi^i]{z}\dd[\varphi^j]{\oz}
=\int_{\Sigma}\varphi^*\omega
\end{align*}
remains, thus proving (\ref{eqnLw}).
\end{proof}

\begin{proof}[Proof of Thm. \ref{thmSupersymmetricAction}]
Lem. \ref{lemNonsuperAction} identifies the action functional $\mA_1$ with $\int_{\Sigma}\mL$
which, in turn, decomposes into two terms by the super action identity from
Prp. \ref{prpSuperActionIdentity}. Consider a variation $\Phi_{\varepsilon}$ of $\Phi$ such that
$\dbar_J\Phi_0=0$. We already know that the derivative of the first term
$\frac{d}{d\varepsilon}|_0\int_{\Sigma}\varphi_{\varepsilon}^*\omega=0$
vanishes since the integral is constant for $\varphi_{\varepsilon}$ within a fixed homology class.
Hence, only the second term remains, and we calculate
\begin{align*}
&\frac{i}{2}\frac{d}{d\varepsilon}|_0\int_{\Sigma}\mL(\Phi)\\
&\qquad=\frac{d}{d\varepsilon}|_0\int_{\Sigma}dz\wedge d\oz\,\dd{\theta^+}\dd{\theta^-}\,\left((\phi_{\varepsilon}\circ g_{ij})
\cdot(\dbar_J\Phi_{\varepsilon}(D_+))(x^i)\cdot(\dbar_J\Phi_{\varepsilon}(D_-))(x^j)\right)\\
&\qquad=\int_{\Sigma}dz\wedge d\oz\,\dd{\theta^+}\dd{\theta^-}\,\left(\frac{d}{d\varepsilon}|_0(\phi_{\varepsilon}\circ g_{ij})\cdot
(\dbar_J\Phi_0(D_+))(x^i)\cdot(\dbar_J\Phi_0(D_-))(x^j)\right.\\
&\qquad\qquad\qquad\qquad\qquad\qquad\quad
+(\phi_0\circ g_{ij})\cdot\frac{d}{d\varepsilon}|_0(\dbar_J\Phi_{\varepsilon}(D_+))(x^i)\cdot(\dbar_J\Phi_0(D_-))(x^j)\\
&\qquad\qquad\qquad\qquad\qquad\qquad\quad
\left.+(\phi_0\circ g_{ij})\cdot(\dbar_J\Phi_0(D_+))(x^i)\cdot\frac{d}{d\varepsilon}|_0(\dbar_J\Phi_{\varepsilon}(D_-))(x^j)\right)
\end{align*}
A more elegant variant of this calculation
involves the (super-)pullback of the Levi-Civita connection of $g$ which is still metric.
Either way, we see that each term in the integral vanishes by $\dbar_J\Phi_0=0$.
Therefore, holomorphic supercurves extremise $\mA_1$.
\end{proof}

\end{document}